\documentclass[runningheads,envcountsame,envcountsect,12pt]{llncs}
\usepackage[a4paper,margin=0.75in]{geometry}

\usepackage{graphbox,graphicx}
\usepackage{subcaption}

\usepackage{hyperref,xcolor}

\usepackage{amsmath,amssymb}

\usepackage{mathtools}

\spnewtheorem{observation}[theorem]{Observation}{\bfseries}{\itshape}
\spnewtheorem{myclaim}[theorem]{Claim}{\bfseries}{\itshape}

\newcommand{\fig}{./fig}

\usepackage{color}
\definecolor{darkred}{rgb}{0.7,0,0}

\definecolor{lightblue}{rgb}{.3,.3,1}

\usepackage{tabularx,environ}

\makeatletter
\newcommand{\problemtitle}[1]{\gdef\@problemtitle{#1}}
\newcommand{\probleminput}[1]{\gdef\@probleminput{#1}}
\newcommand{\problemquestiontitle}[1]{\gdef\@problemquestiontitle{#1}}
\newcommand{\problemquestion}[1]{\gdef\@problemquestion{#1}}

\NewEnviron{myproblem}{
  \problemtitle{}\probleminput{}\problemquestion{}
  \BODY
  \par\addvspace{.5\baselineskip}
  \noindent \normalsize
  \begin{tabularx}{\textwidth}{@{\hspace{\parindent}} l X c}
    \multicolumn{2}{@{\hspace{\parindent}}l}{\normalsize\@problemtitle} \\
    \normalsize \textbf{Input:} & \normalsize \  \@probleminput \\
    \normalsize \textbf{\@problemquestiontitle:} & \normalsize \  \@problemquestion
  \end{tabularx}
  \par\addvspace{.5\baselineskip}
}
\makeatother

\newcommand{\alpp}{\textsf{ALPP}}
\newcommand{\falpp}{\textsf{Full}-\alpp}
\newcommand{\xalpp}{\textsf{Extended}-\alpp}
\newcommand{\pcsat}{\textsc{Planar Circuit SAT}}
\newcommand{\bsat}{3-\textsc{Sat}($2, 1$)}

\newcommand{\sapp}{\textsf{SAPP}}
\newcommand{\fsapp}{\textsf{Full}-\sapp}

\newcommand{\pw}{\mathsf{pw}} 
\newcommand{\ns}{\mathsf{ns}} 
\newcommand{\tw}{\mathsf{tw}} 
\newcommand{\cw}{\mathsf{cw}} 

\newcommand{\figref}[1]{\figurename~\ref{#1}}

\renewcommand{\orcidID}[1]{} 

\begin{document}


\title{Parameterized Complexity of $(A,\ell)$-Path Packing\thanks{%
Partially supported
by PRC CNRS JSPS project PARAGA,
by JSPS KAKENHI Grant Numbers 
JP17H01698, 
JP18H04091, 
JP18H05291, 
JP18K11157, 
JP18K11168, 
JP18K11169, 
JP19K21537, 
JP20K11692, 
JP20K19742. 
The authors thank Tatsuya Gima for helpful discussions.
A preliminary version appeared in the proceedings of
the 31st International Workshop on Combinatorial Algorithms (IWOCA 2020),
Lecture Notes in Computer Science 12126 (2020) 43--55.}}
%

\author{%
R\'{e}my Belmonte\inst{1}\and 
Tesshu Hanaka\inst{2}\orcidID{0000-0001-6943-856X} \and
Masaaki Kanzaki\inst{3} \and
Masashi Kiyomi\inst{4}\orcidID{} \and
Yasuaki~Kobayashi\inst{5}\orcidID{} \and 
Yusuke Kobayashi\inst{5}\orcidID{0000-0001-9478-7307} \and 
Michael Lampis\inst{6}\orcidID{0000-0002-5791-0887} \and 
Hirotaka Ono\inst{7}\orcidID{0000-0003-0845-3947} \and
Yota~Otachi\inst{7}\orcidID{0000-0002-0087-853X}
}
\authorrunning{Belmonte et al.}
%

\institute{
The University of Electro-Communications, Chofu, Tokyo, Japan\\
\email{remybelmonte@gmail.com}
\and
Chuo University, Bunkyo-ku, Tokyo, Japan\\
\email{hanaka.91t@g.chuo-u.ac.jp}
\and
Japan Advanced Institute of Science and Technology, Nomi, Japan\\
\email{kanzaki@jaist.ac.jp}
\and
Yokohama City University, Yokohama, Japan\\
\email{masashi@yokohama-cu.ac.jp}
\and
Kyoto University, Kyoto, Japan\\
\email{kobayashi@iip.ist.i.kyoto-u.ac.jp}, \email{yusuke@kurims.kyoto-u.ac.jp}
\and
Universit\'{e} Paris-Dauphine, PSL University, CNRS, LAMSADE, 75016, Paris, France\\
\email{michail.lampis@lamsade.dauphine.fr}
\and
Nagoya University, Nagoya, 464-8601, Japan\\
\email{ono@nagoya-u.jp}, \email{otachi@nagoya-u.jp}
}

\maketitle

\begin{abstract}
Given a graph $G = (V,E)$, $A \subseteq V$, and integers $k$ and $\ell$,
the \textsc{$(A,\ell)$-Path Packing} problem asks to find $k$ vertex-disjoint paths of length $\ell$
that have endpoints in $A$ and internal points in $V \setminus A$.
We study the parameterized complexity of this problem 
with parameters $|A|$, $\ell$, $k$, treewidth, pathwidth,
and their combinations.
We present sharp complexity contrasts with respect to these parameters.
Among other results, we show 
that the problem is polynomial-time solvable when $\ell \le 3$,
while it is NP-complete for constant $\ell \ge 4$.
We also show that the problem is W[1]-hard parameterized by pathwidth${}+|A|$,
while it is fixed-parameter tractable parameterized by treewidth${}+\ell$.

\end{abstract}


\section{Introduction}
\label{sec:intro}

Let $G = (V,E)$ be a graph and $A \subseteq V$.
A path $P$ in $G$ is an \emph{$A$-path} if the first and the last vertices of $P$ belong to $A$
and all other vertices of $P$ belong to $V \setminus A$.
Given $G$ and $A$, \textsc{$A$-Path Packing} is 
the problem of finding the maximum number of vertex-disjoint $A$-paths in $G$.
The \textsc{$A$-Path Packing} problem 
is well studied and even some generalized versions are known to be polynomial-time solvable
(see e.g.,~\cite{Gallai61,Mader78,ChudnovskyGGGLS06,Pap07,ChudnovskyCG08,Pap08}).
Note that \textsc{$A$-Path Packing}
is a generalization of \textsc{Maximum Matching}
since they are equivalent when $A = V$.

In this paper, we study a variant of \textsc{$A$-Path Packing} that also generalizes \textsc{Maximum Matching}.
An $A$-path of length $\ell$ is an \emph{$(A,\ell)$-path},
where the length of a path is the number of edges in the path.
Now our problem is defined as follows:
\begin{myproblem}
  \problemtitle{\textsc{$(A,\ell)$-Path Packing} \ (\alpp)}
  \probleminput{A tuple $(G,A,k,\ell)$, where $G = (V,E)$ is a graph, $A \subseteq V$, 
    and $k$ and $\ell$ are positive integers.}
  \problemquestiontitle{Question}
  \problemquestion{Does $G$ contain $k$ vertex-disjoint $(A,\ell)$-paths?}
\end{myproblem}

To the best of our knowledge, this natural variant of \textsc{$A$-Path Packing} was not studied in the literature.
Our main motivation of studying {\alpp} is to see theoretical differences from the original \textsc{$A$-Path Packing},
but practical motivations of having the length constraint may come from some physical restrictions
or some fairness requirements.
Note that if $\ell = 1$, then {\alpp} is equivalent to \textsc{Maximum Matching}.
Another related problem is \textsc{$\ell$-Path Partition}~\cite{YanCHH97,Steiner03,MonnotT07},
which asks for vertex-disjoint paths of length $\ell$ (without specific endpoints).

In the rest of paper, we assume that $k \le |A|/2$ in every instance 
as otherwise the instance is a trivial no-instance.
The restricted version of the problem where the equality $k = |A|/2$ is forced
is also of our interest as that version corresponds to a ``full'' packing of $A$-paths.
We call this version \textsc{Full $(A,\ell)$-Path Packing} (\falpp, for short).
In this paper,
all our positive results showing tractability of some cases will be on the general {\alpp},
while all our negative (or hardness) results will be on the possibly easier {\falpp}.

We assume that the reader is familiar with terminologies
in the parameterized complexity theory.
See the textbook by Cygan et al.~\cite{CyganFKLMPPS15} for standard definitions.

\subsection*{Our results}
In summary, we show that {\alpp} is intractable even on very restricted inputs,
while it has some nontrivial cases that admit efficient algorithms.
(See \figref{fig:summary}.)

We call $|A|$, $k$, and $\ell$ the \emph{standard parameters} of {\alpp} as they naturally arise
from the definition of the problem.
We determine the complexity of {\alpp} with respect to all standard parameters
and their combinations.
We first observe that {\falpp} is NP-complete for any constant $|A| \ge 2$ (Observation~\ref{obs:fixed-|A|})
and for any constant $\ell \ge 4$ (Observation~\ref{obs:fixed-ell}),
while it is polynomial-time solvable when $\ell \le 3$ (Theorem~\ref{thm:ell<=3}).
On the other hand, {\alpp} is fixed-parameter tractable when parameterized by $k + \ell$
and thus by $|A| + \ell$ as well (Theorem~\ref{thm:k+ell}).
We later strengthen Observation~\ref{obs:fixed-ell}
by showing that NP-complete for every fixed $\ell \ge 4$ even on grid graphs
(Theorem~\ref{thm:grid}).

We then study structural parameters such as treewidth and pathwidth
in combination with the standard parameters.
We first observe that {\alpp} can be solved in time $n^{O(\tw)}$ (Theorem~\ref{thm:tw-XP}), 
where $n$ and $\tw$ are the number of vertices and the treewidth of the input graph, respectively.
Furthermore, we show that {\alpp} parameterized by $\tw + \ell$ is fixed-parameter tractable (Theorem~\ref{thm:tw+ell}).
We finally show that {\falpp} parameterized by $\pw + |A|$ is W[1]-hard (Theorem~\ref{thm:pw+|A|}),
where $\pw$ is the pathwidth of the input graph.

We also study a variant of the problem in which we are allowed to use $A$-paths shorter than $\ell$ as well.
A simple reduction (Lemma~\ref{lem:sapp-to-alpp}) will show that
all the positive results on {\alpp} can be translated to the ones on this ``short $A$-path'' variant.
Although negative results cannot be translated directly,
we can show the hardness of the cases where $|A|$ or $\ell$ is a constant.
We leave the complexity of this variant parameterized by $\tw$ unsettled.

\begin{figure}[bth]
  \centering
  \includegraphics[width=.8\linewidth]{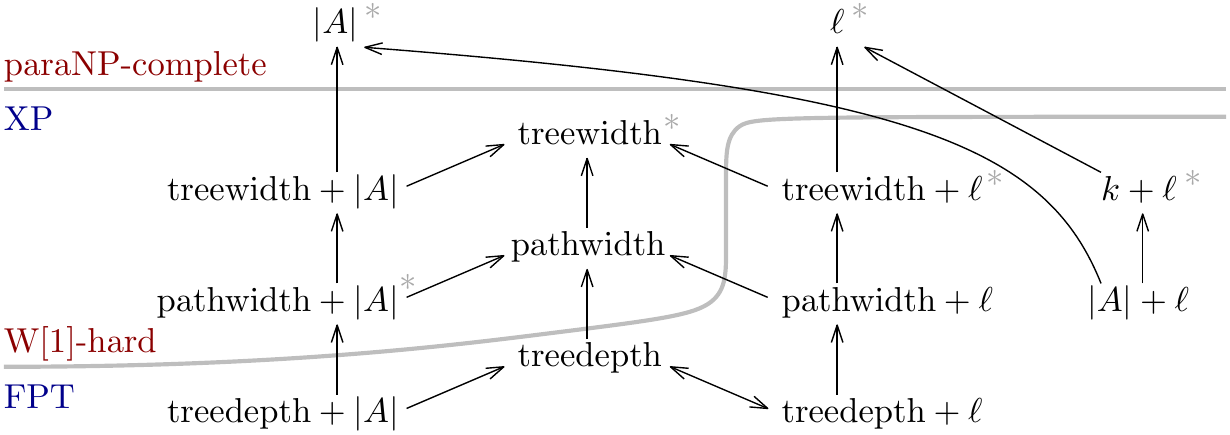}
  \caption{Summary of the results. An arrow $\alpha \rightarrow \beta$ indicates that there is a function $f$ such that $\alpha \ge f(\beta)$ for every instance of {\alpp}.
    Some possible arrows are omitted to keep the figure readable.
    The results on the parameters marked with $\ast$ are explicitly shown in this paper,
    and the other results follow by the hierarchy of the parameters.
    We have a bidirectional arrow $\text{treedepth} \leftrightarrow \text{treedepth} + \ell$
    because the maximum length of a path in a graph is bounded by a function of treedepth~\cite[Section 6.2]{NesetrilDM12}.}
  \label{fig:summary}
\end{figure}

\section{Preliminaries}
\label{sec:pre}

A graph $G = (V,E)$ is a \emph{grid graph} if 
$V$ is a finite subset of $\mathbb{Z}^{2}$
and $E = \{\{(r,c), (r',c')\} \mid |r-r'| + |c-c'| = 1\}$.
From the definition, all grid graphs are planar, bipartite, and of maximum degree at most 4.
To understand the intractability of a graph problem, 
it is preferable to show hardness on a very restricted graph class.
The class of grid graphs is one of such target classes.

A \emph{tree decomposition} of a graph $G = (V,E)$ is a pair $(\{X_{i} \mid i \in I\}, T=(I,F))$,
where $X_{i} \subseteq V$ for each $i$ and $T$ is a tree such that 
\begin{itemize}
  \item for each vertex $v \in V$, there is $i \in I$ with $v \in X_{i}$;
  \item for each edge $\{u,v\} \in E$, there is $i \in I$ with $u, v \in X_{i}$;
  \item for each vertex $v \in V$, the induced subgraph $T[\{i \mid v \in X_{i}\}]$ is connected.
\end{itemize}
The \emph{width} of a tree decomposition $(\{X_{i} \mid i \in I\}, T)$ is $\max_{i \in I} |X_{i}| - 1$,
and the \emph{treewidth} of a graph $G$, denoted $\tw(G)$, is the minimum width over all tree decompositions of $G$.

The \emph{pathwidth} of a graph $G$, denoted $\pw(G)$, is defined by restricting the trees $T$ in tree decompositions to be paths.
We call such decompositions \emph{path decompositions}.

We can show that pathwidth does not change significantly
by subdividing some edges and attaching paths to some vertices.
To this end, we use a characterization of pathwidth by the following search game.
We are given a graph $G = (V,E)$ with all edges \emph{contaminated}.
The goal in this game is to clear all edges.
In each turn, we can place a \emph{searcher} on a vertex or remove a searcher from a vertex.
An edge is \emph{cleared} by having searchers on both endpoints.
A cleared edge is immediately \emph{recontaminated} when a removal of a searcher
results in a path not passing through any searchers from the edge to a contaminated edge.
The minimum number of searchers needed to clear all edges of $G$ is the \emph{node search number},
and we denote it by $\ns(G)$.
It is known that $\ns(G) = \pw(G)+1$ for every graph $G$~\cite{KirousisP85,EllisST94,Bodlaender98}.

\begin{lemma}
  [Folklore]
  \label{lem:pw-subdivision}
  Let $G = (V,E)$ be a graph.
  If $G'$ is a graph obtained from $G$ by subdividing a set of edges $F \subseteq E$ an arbitrary number of times,
  then $\pw(G') \le \pw(G) + 2$.
\end{lemma}
\begin{proof}
Let $p = \pw(G) + 1$.
Since $\ns(G) \le p$,
there is a sequence $S$ of placements and deletions of searchers to clear all edges of $G$ using at most $p$ searchers.
To clear all edges of $G'$, we extend $S$ as follows.
For each placement of a searcher on a vertex $v \in V$,
we insert, right after this placement,
a subsequence that clears all paths corresponding to the edges between $v$ and its neighbors having searchers on them at this point.
This can be done with two extra searchers that clear the paths one-by-one
The extra searchers are deleted in the end of the subsequence,
and thus we only need two extra searchers in total. 
This implies that $\pw(G') = \ns(G')-1 \le \ns(G)+1 = \pw(G) + 2$. 
\qed
\end{proof}

\begin{lemma}
  [Folklore]
  \label{lem:pw-path-attachment}
  If $G'$ is a graph obtained from a graph $G = (V,E)$
  by attaching a path of arbitrary length to each vertex in a set $U \subseteq V$,
  then $\pw(G') \le \pw(G) + 1$.
\end{lemma}
\begin{proof}
We assume that $E \ne \emptyset$ since otherwise the statement is clearly true.
There is a sequence $S$ of placements and deletions of searchers to clear all edges of $G$ using at most $p = \pw(G) + 1$ searchers.
To clear all edges of $G'$, we extend $S$ as follows.
Using two searchers, we first clear the paths attached to the isolated vertices in $U$ (if such exist).
We then replace each placement of a searcher on a non-isolated vertex $u \in U$
with a subsequence that clears the path $P = (u,p_{1},p_{2},\dots,p_{q})$ attached to $u$.
This can be done with two searchers by first placing a searcher on $p_{q}$, then placing a searcher on $p_{q-1}$,
deleting a searcher on $p_{q}$, placing a searcher on $p_{q-2}$, and so on.
At the end of the subsequence, $u$ has a searcher on it and all other vertices in $P$ do not have searchers.
We need only one extra searcher in total. 
Hence, $\pw(G') = \ns(G')-1 \le \ns(G) = \pw(G) + 1$.
\qed
\end{proof}
In the proofs of Lemmas~\ref{lem:pw-subdivision} and \ref{lem:pw-path-attachment},
we show that search sequences for the original graph can be ``locally'' extended for the new graph
by using one or two temporal searchers. Thus we can have the following combined version of the lemmas
as Corollary~\ref{cor:pw-subdivision-path}.
\begin{corollary}
\label{cor:pw-subdivision-path}
If $G'$ is a graph obtained from a graph $G = (V,E)$
 by subdividing a set of edges $F \subseteq E$ an arbitrary number of times,
and attaching a path of arbitrary length to each vertex in a set $U \subseteq V$,
then $\pw(G') \le \pw(G) + 2$.
\end{corollary}

\section{Standard parameterizations of {\alpp}}
\label{sec:standard-parameters}

In this section, we completely determine the complexity of {\alpp} 
with respect to the standard parameters $|A|$, $k$, $\ell$, and their combinations.
(Recall that $k \le |A|/2$.)
We first observe that using one of them as a parameter does not make the problem tractable.
That is, we show that the problem remains NP-complete 
even if one of $|A|$, $k$, $\ell$ is a constant.
We then show that the problem is tractable when $\ell \le 3$ or
when $k+\ell$ is the parameter.

\subsection{Intractable cases}
The first observation is that {\falpp} is NP-complete even if $|A| = 2$ (and thus $k = 1$).
This can be shown by an easy reduction from \textsc{Hamiltonian Cycle}~\cite{GareyJ79}.
This observation is easily extended to every fixed even $|A|$.
\begin{observation}
\label{obs:fixed-|A|}
For every even constant $\alpha \ge 2$,
{\falpp} with $|A| = \alpha$ is NP-complete on grid graphs.
\end{observation}
\begin{proof}
Let $G = (V,E)$ be an instance of \textsc{Hamiltonian Cycle} on grid graphs,
which is known to be NP-complete~\cite{ItaiPS82}.
Since {\falpp} is clearly in NP, 
it suffices to construct an equivalent instance of {\falpp} in polynomial time.

Observe that the minimum degree $\delta(G)$ of $G$ is at most 2.
If $\delta(G) < 2$, then $G$ is a no-instance of \textsc{Hamiltonian Cycle}, 
and thus we can construct trivial no-instance of {\falpp}.
Assume that $\delta(G) = 2$, and let $v$ be a vertex of degree 2 in $G$ with the neighbors $u$ and $w$.
Let $G'$ be the graph obtained from $G$ by removing $v$
and then adding $\alpha/2 - 1$ isolated paths of length $|V|-2$ if $\alpha > 2$.
Let $Q$ be the set of endpoints of the new paths.
Then, $(G', \{u,w\} \cup Q, \alpha/2, |V|-2)$ is a yes-instance of {\falpp} if and only if $G$ has a Hamiltonian cycle.
This can be seen by observing that each Hamiltonian cycle of $G$ includes edges $\{u,v\}$ and $\{v,w\}$
and that each $(\{u,w\}, |V|-2)$-path in $G'$ can be extended to a Hamiltonian cycle of $G$ by using $v$ and
the edges $\{u,v\}$ and $\{v,w\}$.
\qed
\end{proof}

The NP-hardness of {\falpp} for fixed $\ell$ can be shown also by an easy reduction from a known NP-hard problem,
but in this case only for $\ell \ge 4$. 
This is actually tight as we see later that the problem is polynomial-time solvable when $\ell \le 3$ (see Theorem~\ref{thm:ell<=3}).
\begin{observation}
\label{obs:fixed-ell}
For every constant $\ell \ge 4$, {\falpp} is NP-complete.
\end{observation}
\begin{proof}
Given a graph $G = (V,E)$,
the \textsc{$\lambda$-Path Partition} problem asks
whether $G$ contains $k := |V|/(\lambda+1)$ vertex-disjoint paths of length $\lambda$.
For every fixed $\lambda \ge 2$, \textsc{$\lambda$-Path Partition} is NP-complete~\cite{KirkpatrickH83}.
We construct $G'$ from $G$ by adding a set $A$ of $2k$ new vertices,
where $A$ is an independent set in $G'$
and $G'$ has all possible edges between $A$ and $V$.
We can see that $G$ is a yes-instance of \textsc{$\lambda$-Path Partition}
if and only if
$(G', A, k, \ell)$ is a yes-instance of {\falpp}, where $\ell = \lambda+2 \ge 4$.
\qed
\end{proof}

We can strengthen Observation~\ref{obs:fixed-ell} to hold on grid graphs
by constructing an involved reduction from scratch.
As the proof is long and the theorem does not really fit the theme of this section,
we postpone it to Section~\ref{sec:grid-graphs}.

\subsection{Tractable cases}

\begin{theorem}
\label{thm:ell<=3}
If $\ell \le 3$,
then {\alpp} can be solved in polynomial time.
\end{theorem}
\begin{proof}
Let $(G,A,k,\ell)$ with $G = (V,E)$ be an instance of {\alpp} with $\ell \le 3$.

If $\ell = 1$, then the problem can be solved by finding a maximum matching in $G[A]$.
Since a maximum matching can be found in polynomial time~\cite{Edmonds65},
this case is polynomial-time solvable.

Consider the case where $\ell = 2$. 
We reduce this case to the case of $\ell = 3$.
We can assume that $G[A]$ and $G[V \setminus A]$ do not contain any edges
as such edges are not included in any $(A,2)$-path.
New instance $(G',A,k,3)$ is constructed
by adding a true twin $v'$ to each vertex $v \in V \setminus A$;
i.e., $V(G') = V \cup \{v' \mid v \in V \setminus A\}$
and $E(G') = E \cup \{\{v,v'\} \mid v \in V \setminus A \} \cup \{\{u,v'\} \mid u \in A, v \in V \setminus A, \{u,v\} \in E\}$.
Clearly, $(G,A,k,2)$ is a yes-instance if and only if so is $(G',A,k,3)$.

\medskip

For the case of $\ell = 3$, we construct an auxiliary graph $G' = (A \cup V_{1} \cup V_{2}, E_{A,1} \cup E_{1,2} \cup E_{2,2})$ as follows
(see \figref{fig:l3}):
\begin{align*}
  V_{i} &= \{v_{i} \mid v \in V \setminus A\} \text{ for } i \in \{1,2\},
  \\
  E_{A,1} &= \{\{u, v_{1}\} \mid u \in A, \; v \in V \setminus A, \; \{u,v\} \in E\},
  \\
  E_{1,2} &= \{\{v_{1}, v_{2}\} \mid v \in V\},
  \\
  E_{2,2} &= \{\{u_{2},v_{2}\} \mid u,v \in V \setminus A, \; \{u,v\} \in E\}.
\end{align*}
We show that $(G,A,k,3)$ is a yes-instance if and only if $G'$ has a matching of size $k + |V \setminus A|$,
which implies that the problem can be solved in polynomial time.

\begin{figure}[bth]
  \centering
  \includegraphics[width=.9\linewidth]{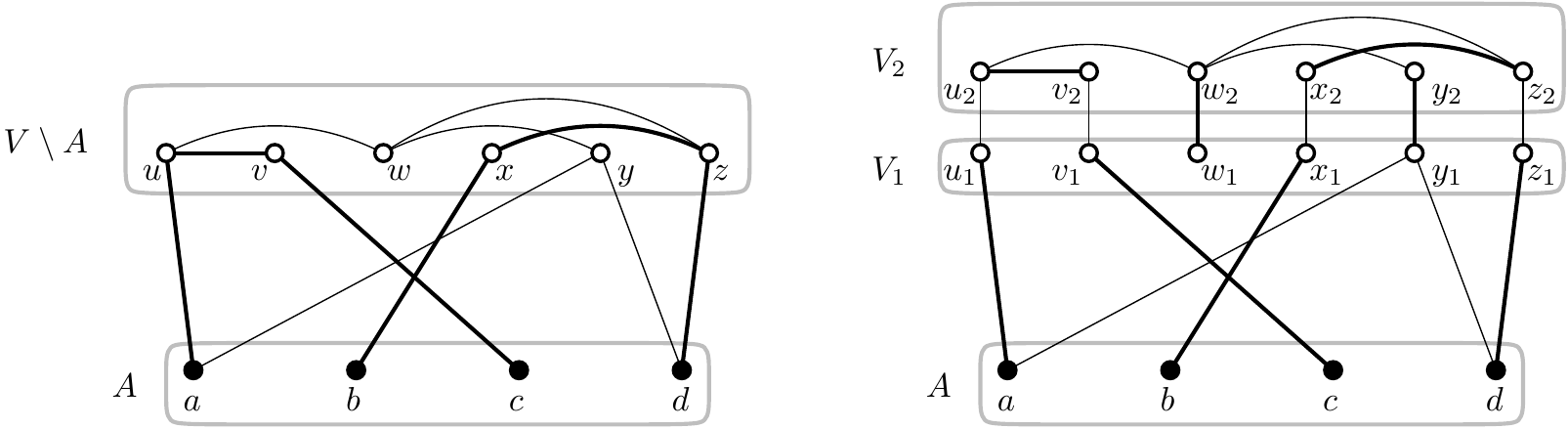}
  \caption{The construction of $G'$ (right) from $G$ (left).}
  \label{fig:l3}
\end{figure}

To prove the only-if direction,
let $P_{1}, \dots, P_{k}$ be $k$ vertex-disjoint $(A,3)$-path in $G$.
We set $M = M_{A,1} \cup M_{1,2} \cup M_{2,2}$, where
\begin{align*}
  M_{A,1} &= \{\{u, v_{1}\} \in E_{A,1} \mid \text{edge } \{u,v\} \text{ appears in some } P_{i}\},
  \\
  M_{1,2} &= \{\{v_{1}, v_{2}\} \in E_{1,2} \mid \text{vertex } v \text{ does not appear in any } P_{i}\},
  \\
  M_{2,2} &= \{\{u_{2},v_{2}\} \in E_{2,2} \mid \text{edge } \{u,v\} \text{ appears in some } P_{i}\}.
\end{align*}
Since the $(A,3)$-paths $P_{1}, \dots, P_{k}$ are pairwise vertex-disjoint,
$M$ is a matching. We can see that $|M| = k + |V \setminus A|$ 
as $|M_{2,2}| = k$ and $|M_{A,1}| + |M_{1,2}| = |V_{1}| = |V \setminus A|$.

\smallskip

To prove the if direction, assume that $G'$ has a matching of size $k+|V \setminus A|$.
Let $M$ be a maximum matching of $G'$ that includes the maximum number of vertices
in $V_{1} \cup V_{2}$ among all maximum matchings of $G'$.
We claim that $M$ actually includes all vertices in $V_{1} \cup V_{2}$.
Suppose to the contrary that $v_{1}$ or $v_{2}$ is not included in $M$ for some $v \in V \setminus A$.
Now, since $M$ is maximum, exactly one of $v_{1}$ and $v_{2}$ is included in $M$.

Case 1: $v_{1} \in V(M)$ and $v_{2} \notin V(M)$.
There is a vertex $u \in A$ such that $\{u, v_{1}\} \in M$.
The set $M - \{u, v_{1}\} + \{v_{1}, v_{2}\}$ is a maximum matching 
that uses more vertices in $V_{1} \cup V_{2}$ than $M$.
This contradicts how $M$ was selected.

Case 2: $v_{1} \notin V(M)$ and $v_{2} \in V(M)$.
There is a vertex $w_{2} \in V_{2}$ such that $\{v_{2}, w_{2}\} \in M$.
The edge set $M' := M - \{v_{2}, w_{2}\} + \{v_{1}, v_{2}\}$ is a maximum matching 
that uses the same number of vertices in $V_{1} \cup V_{2}$ as $M$.
Since $M'$ is maximum and $w_{2}$ is not included in $M'$, 
the vertex $w_{1}$ has to be included in $M'$,
but such a case leads to a contradiction as we saw in Case 1.

Now we construct $k$ vertex-disjoint $(A,3)$-paths from $M$ as follows.
Let $\{u_{2}, v_{2}\} \in M \cap E_{2,2}$.
Since $M$ includes all vertices in $V_{1}$,
it includes edges $\{u_{1}, x\}$ and $\{v_{1}, y\}$ for some $x, y \in A$.
This implies that $G$ has an $(A,3)$-path $(x, u, v, y)$.
Let $(x', u', v', y')$ be the $(A,3)$-path constructed in the same way from a different edge in $M \cap E_{2,2}$.
Since $M$ is a matching, these eight vertices are pairwise distinct,
and thus $(x, u, v, y)$ and $(x', u', v', y')$ are vertex-disjoint $(A,3)$-paths.
Since $|M| \ge k + |V \setminus A|$ and each edge in $E_{A,1} \cup E_{1,2}$ uses one vertex of $V_{1}$,
$M$ includes at least $k$ edges in $E_{2,2}$.
By constructing an $(A,3)$-path for each edge in $M \cap E_{2,2}$,
we obtain a desired set of $k$ vertex-disjoint $(A,3)$-paths.
\qed
\end{proof}

In their celebrated paper on \emph{Color-Coding}~\cite{AlonYZ95},
Alon, Yuster, and Zwick showed the following result.
\begin{proposition}
[{\cite[Theorem 6.3]{AlonYZ95}}]
\label{prop:color-coding}
Let $H$ be a graph on $h$ vertices with treewidth $t$.
Let $G$ be a graph on $n$ vertices. 
A subgraph of $G$ isomorphic to $H$, if one exists, can be found in time
$O(2^{O(h)} \cdot n^{t+1} \log n)$.
\end{proposition}

By using Proposition~\ref{prop:color-coding} as a black box,
we can show that {\alpp} parameterized by $k + \ell$ is fixed-parameter tractable.
\begin{theorem}
\label{thm:k+ell}
{\alpp} on $n$-vertex graphs can be solved in $O(2^{O(k \ell)} n^{6} \log n)$ time.
\end{theorem}
\begin{proof}
Let $(G,A,k,\ell)$ be an instance of {\alpp}.
Observe that the problem {\alpp} can be seen as a variant of the \textsc{Subgraph Isomorphism} problem
as we search for $H = k P_{\ell+1}$ in $G$ as a subgraph with the restriction that each endpoint of $P_{\ell+1}$ in $H$
has to be mapped to a vertex in $A$,
where $P_{\ell+1}$ denotes an $(\ell+1)$-vertex path (which has length $\ell$)
and $k P_{\ell+1}$ denotes the disjoint union of $k$ copies of $P_{\ell+1}$.
We reduce this problem to the standard \textsc{Subgraph Isomorphism} problem~\cite{GareyJ79}.

Let $G'$ and $H'$ be the graphs obtained from $G$ and $H$, respectively, by subdividing each edge once.
The graphs $G'$ and $H' = k P_{2\ell + 1}$ are bipartite.
We then construct $G''$ from $G'$ by attaching a triangle to each vertex in $A$;
that is, for each vertex $u \in A$ we add two new vertices $v, w$ and edges $\{u,v\}$, $\{v,w\}$, and $\{w,u\}$.
Similarly, we construct $H''$ from $H'$ by attaching a triangle to each endpoint of each $P_{2\ell+1}$.
Note that $|V(G'')| \in O(n^{2})$, $|V(H'')| = k(2 \ell + 1)$, and $\tw(H'') = 2$.
Thus, by Proposition~\ref{prop:color-coding}, it suffices to show that
$(G,A,k,\ell)$ is a yes-instance of {\alpp} if and only if $G''$ has a subgraph isomorphic to $H''$.

\smallskip

To show the only-if direction,
assume that $G$ has $k$ vertex-disjoint $(A,\ell)$-paths $P_{1}, \dots, P_{k}$.
In $G''$, for each $P_{i}$, there is a unique path $Q_{i}$ of length $2\ell$ plus triangles attached to the endpoints;
that is, $Q_{i}$ consists of the vertices of $P_{i}$, the new vertices and edges introduced by subdividing the edges
in $P_{i}$, and the triangles attached to the endpoints of the subdivided path.
Furthermore, since the paths $P_{i}$ are pairwise vertex-disjoint,
the subgraphs $Q_{i}$ of $G''$ are pairwise vertex-disjoint.
Thus, $G''$ has a subgraph isomorphic to $H'' = \bigcup_{1 \le i \le k} Q_{i}$.

\smallskip

To prove the if direction,
assume that $G$ has a subgraph $H'$ isomorphic to $H$.
Let $R_{1}, \dots, R_{k}$ be the connected components of $H'$.
Each $R_{i}$ is isomorphic to a path of length $2\ell$ with a triangle attached to each endpoint.
Let $u, v \in V(R_{i})$ be the degree-3 vertices of $R_{i}$.
Since $G''$ is obtained from the triangle-free graph $G'$ by attaching triangles at the vertices in $A$, 
we have $u, v \in A$. 
Since the $u$-$v$ path of length $2\ell$ in $R_{i}$ is obtained from 
a $u$-$v$ path of length $\ell$ in $G$ by subdividing each edge once,
the graph $G[V(R_{i}) \cap V(G)]$ contains an $(A,\ell)$-path.
Since $V(R_{1}), \dots, V(R_{k})$ are pairwise disjoint,
$G$ contains $k$ vertex-disjoint $(A,\ell)$-paths.
\qed
\end{proof}


\section{Structural parameterizations}
\label{sec:structural-parameters}

In this section, we study structural parameterizations of {\alpp}.
First we present XP and FPT algorithms parameterized by $\tw$ and $\tw + \ell$, respectively.

The XP-time algorithm parameterized by $\tw$ is based on 
an efficient algorithm for computing a tree decomposition~\cite{BodlaenderDDFLP16}
and a standard dynamic-programming over \emph{nice tree decompositions}~\cite{Kloks94}.
The FPT algorithm parameterized $\tw + \ell$ is achieved by expressing the problem in 
the \emph{monadic second-order logic} (MSO$_{2}$) of graphs~\cite{ArnborgLS91,Courcelle92,Bodlaender96}.
\begin{theorem}
\label{thm:tw-XP}
{\alpp} can be solved in time $n^{O(\tw)}$.
\end{theorem}
\begin{proof}
Let $G$ be an $n$-vertex graph of treewidth at most $\tw$.
We compute the maximum number of vertex-disjoint $(A,\ell)$-paths
by a standard dynamic programming algorithm over a tree decomposition.
To this end, it is helpful to use so called nice tree decompositions~\cite{Kloks94}.
A tree decomposition $(\{X_{i} \mid i \in I\}, T=(I,F))$ is \emph{nice}
if $T$ is a rooted tree, each node of $T$ has at most two children, and
\begin{itemize}
  \item if $i \in I$ has exactly one child $j$,
  then $X_{i} = X_{j} \cup \{u\}$ for some $u \notin X_{j}$ or $X_{i} = X_{j} \setminus \{v\}$ for some $v \in X_{j}$;
  \item if $i \in I$ has exactly two children $j$ and $h$, then $X_{i} = X_{j} = X_{h}$.
\end{itemize}

We compute a tree decomposition of width 
at most $w = 5\tw+4$ in time $2^{O(\tw)} n$~\cite{BodlaenderDDFLP16},
and then convert it in linear time to a nice tree decomposition $(\{X_{i} \mid i \in I\}, T=(I,F))$ 
of the same width having $O(n)$ nodes in the tree $T$~\cite{Kloks94}.

Let $V_{i} = \bigcup_{j} X_{j}$, where the union is taken over all descendants $j$ of $i$ in $T$ (including $i$ itself).
For each $i \in I$, we define the DP table $\textsf{dp}_{i}(\alpha,\lambda,\delta,\kappa) \in \{\textsf{true}, \textsf{false}\}$
with the indices 
$\alpha\colon X_{i} \to \{B \subseteq A \mid |B| \le 2\}$,
$\lambda\colon X_{i} \to \{0,\dots,\ell\}$,
$\delta\colon X_{i} \to \{0, 1, 2\}$,
$\kappa \in \{0,\dots,|A|/2\}$
such that $\textsf{dp}_{i}(\alpha,\lambda,\delta,\kappa) = \textsf{true}$
if and only if there exists a spanning subgraph $H$ of $G[V_{i}]$ such that all the following conditions are satisfied:
\begin{itemize}
  \item all connected components of $H$ are paths, and $\kappa$ of them are $(A,\ell)$-paths;
  \item each vertex in $A \cap V_{i}$ has degree at most 1 in $H$;
  \item for each $v \in X_{i}$, $\alpha(v) = V(C(v)) \cap A$, where $C(v)$ is the connected component of $H$ containing $v$;
  \item $C(v)$ contains exactly $\lambda(v)$ edges for each $v \in X_{i}$; 
  \item each $v \in X_{i}$ has degree $\delta(v)$ in $H$.
\end{itemize}
The size of the table $\textsf{dp}_{i}$ is $O(|A|^{2 (w+1)} \cdot (\ell+1)^{w+1} \cdot 3^{w+1} \cdot |A|/2)$.
Since $|A|$ and $\ell$ are at most $n$, this table size can be bounded by $n^{O(\tw)}$.
If we know all table entries for the root $r$ of $T$,
then we can find the maximum number of vertex-disjoint $(A,\ell)$-paths in $G$ in time $n^{O(\tw)}$,
by just finding the maximum number $\kappa$ such that there exist $\alpha$, $\lambda$, and $\delta$
with $\textsf{dp}_{i}(\alpha,\lambda,\delta,\kappa) = \textsf{true}$.

It is trivial to compute all the entries for a node $i$ with no children in time $n^{O(\tw)}$.
For a node $i$ with one or two children,
if the table entries for the children are already computed,
then it is straightforward to compute  
the table entries for $i$ in time polynomial in the total table size of the children.
This running time is again $n^{O(\tw)}$.
Since there are $O(n)$ nodes in $T$, the total running time is $n^{O(\tw)}$.
\qed
\end{proof}

\begin{theorem}
\label{thm:tw+ell}
{\alpp} parameterized by $\tw + \ell$ is fixed-parameter tractable.
\end{theorem}
\begin{proof}
To show the fixed-parameter tractability of {\alpp} parameterized by $\tw + \ell$,
we use the \emph{monadic second-order logic} (MSO$_{2}$) of graphs.
In an MSO$_{2}$ formula, we can use
(i) the logical connectives $\lor$, $\land$, $\lnot$, $\Leftrightarrow$, $\Rightarrow$,
(ii) variables for vertices, edges, vertex sets, and edge sets,
(iii) the quantifiers $\forall$ and $\exists$ applicable to these variables, and
(iv) the following binary relations:
\begin{itemize}
  \item $u \in U$ for a vertex variable $u$ and a vertex set variable $U$;
  \item $d \in D$ for an edge variable $d$ and an edge set variable $D$;
  \item $\textsf{inc}(d,u)$ for an edge variable $d$ and a vertex variable $u$,
  where the interpretation is that $d$ is incident with $u$;
  \item equality of variables.
\end{itemize}
In the following expressions, we use some syntax sugars, such as $\ne$ and $\notin$, obviously obtained from the definition of MSO$_{2}$
for ease of presentation.
Also, we follow the convention that in MSO$_{2}$ formulas, the set variables $V$ and $E$ denote the vertex and edges sets of the input graph, respectively.

Let $\varphi$ be a fixed MSO$_{2}$ formula.
It is known that given an $n$-vertex graph of treewidth $w$ and assignments to some free variables of $\varphi$,
one can find in time $O(f(|\varphi| + w) \cdot  n)$, where $f$ is some computable function,
assignments to the rest of free variables that satisfies $\varphi$ and maximizes a given linear function 
in the sizes of the free variables of $\varphi$~\cite{ArnborgLS91,Courcelle92,Bodlaender96}.

We can express a formula $(A,\ell)\textsf{-paths}(F)$ that is true if and only if $F$ is the edge set of a set of $(A,\ell)$-paths as follows:
\begin{align*}
 (A,\ell)\textsf{-paths}(F) := {}&\textsf{paths}(F) \land \ell\textsf{-components}(F) \land (\forall v \in V \; (\textsf{deg}_{= 1}(v,F) \implies v \in A)),
\end{align*}
where $\textsf{paths}(F)$ is true if and only if $F$ is the edge set of a set of paths,
$\ell\textsf{-components}(F)$ is true if and only if $F$ is the edge set of a graph
that only has size-$\ell$ components, and $\textsf{deg}_{= 1}(v,F)$ is true if and only if exactly one edge in $F$ has $v$ as an endpoint.
We can easily express these three subformulas
in such a way that the length of the formula $(A,\ell)\textsf{-paths}(F)$ depends only on $\ell$
as follows.

The following formula $\textsf{deg}_{\le d}(v,D)$ has length depending only on $d$
and is true if and only if at most $d$ edges in $D$ has $v$ as an endpoint.
\[
  \textsf{deg}_{\le d}(v,D) := \nexists e_{1}, \dots, e_{d+1} \in D 
  \left(\bigwedge_{1 \le i < j \le d+1} e_{i} \ne e_{j}\right)
  \land
  \left(\bigwedge_{1 \le i \le d+1} \textsf{inc}(e_{i}, v)\right).
\]
Now it is straightforward to express $\textsf{deg}_{= 1}(v,D)$ and $\textsf{paths}(F)$:
\begin{align*}
  \textsf{deg}_{= 1}(v,D) &:= \textsf{deg}_{\le 1}(v,D) \land \lnot \textsf{deg}_{\le 0}(v,D), \\
  \textsf{paths}(F) &:= \left(\forall v \in V (\textsf{deg}_{\le 2}(v,F))\right) \land \left(\forall C \subseteq F, \exists v \in V (\textsf{deg}_{\le 1}(v,C))\right).
\end{align*}

We can express $\ell\textsf{-components}(F)$ as follows
\[
  \ell\textsf{-components}(F) := \forall C \subseteq F (\textsf{component}(C,F) \implies \textsf{size}_{=\ell}(C)),
\]
where $\textsf{component}(C,F)$ is true if and only if $C$ is the edge set of a connected component
(i.e., an inclusion-wise maximal connected subgraph) of the graph induced by $F$,
and $\textsf{size}_{=\ell}(C))$ is true if and only if $C$ includes exactly $\ell$ edges.

As we allow the expression of $\textsf{size}_{=\ell}(C)$ to have length depending on $\ell$,
it is trivially expressible, e.g., as follows:
\begin{align*}
  \textsf{size}_{=\ell}(C) := {}&\exists e_{1}, \dots, e_{\ell} \in C
  \left(\bigwedge_{1 \le i < j \le \ell} (e_{i} \ne e_{j}) 
    \land \left(\forall e' \left(\bigwedge_{1 \le i \le \ell} (e_{i} \ne e') \implies e' \notin C\right)\right)
  \right).
\end{align*}

Expressing the connectivity of the graph induced by an edge set $C$ is a nice exercise
and well known to have the following solution:
\begin{align*}
  \textsf{connected}(C) 
  &:=
  \exists U \subseteq V 
  (\forall v \in V  (v \in U \iff \exists e \in C (\textsf{inc}(e,v)))) \land{}
  \\
  &\qquad\qquad \quad
  (\forall W \subseteq U (W = U \lor (\exists w \in W, \exists z \notin W,  \textsf{adj}(w,z)))),
\end{align*}
where $\textsf{adj}(w,z) := \exists e \in E (\textsf{inc}(e,w) \land \textsf{inc}(e,z))$.
Using this expression, we can express $\textsf{component}(C, F)$ as follows:
\begin{align*}
  \textsf{component}(C, F) 
  &:= 
  \textsf{connected}(C) 
   \land
  \forall e \in F
  (e \notin C \implies \lnot \textsf{connected}(C \cup \{e\})).
\end{align*}

The formula $(A,\ell)\textsf{-paths}(F)$ has two free variables $A$ and $F$.
We assign (or identify) the terminal vertex set $A$ in the input of {\alpp} to the variable $A$,
and maximize the size of $F$. 
As mentioned above, this can be done in time $O(f(|\varphi| + w) \cdot  n)$
for $n$-vertex graphs of treewidth at most $w$, where $f$ is some computable function.
\qed 
\end{proof}

Now we show that {\falpp} is W[1]-hard parameterized by pathwidth
(and hence also by treewidth), even if we also consider $|A|$ as an additional
parameter. We present a reduction from a W[1]-complete problem $k$-\textsc{Multi-Colored Clique} ($k$-\textsc{MCC})~\cite{FellowsHRV09}, which
goes through an intermediate version of our problem. Specifically, we will
consider a version of {\falpp} with the following modifications: the graph has
(positive integer) edge weights, and the length of a path is the sum of the
weights of its edges; the set $A$ is given to us partitioned into pairs
indicating the endpoints of the sought $A$-paths; for each such pair
the value of $\ell$ may be different.

More formally, {\xalpp} is the following problem: we are given a graph
$G=(V,E)$, a weight function $w \colon E \to \mathbb{Z}^+$, and
$r$ triples $(s_1,t_1,\ell_1)$, $\ldots$, $(s_r,t_r,\ell_r) \in V \times V \times \mathbb{Z}^+$,
where all $s_{i}, t_{i} \in V$ are distinct.
We are asked if there exists a set of $r$ vertex-disjoint paths in $G$ 
such that for each $i \in [r]$\footnote{%
For a positive integer $n$, we denote the set $\{1, 2, \dots, n\}$ by $[n]$.},
the $i$th path in this set has $s_i$ and $t_i$ as its endpoints
and the sum of the weights of its edges is $\ell_{i}$.
We first show that establishing that
this variation of the problem is hard implies also the hardness of {\falpp}.

\begin{lemma}
\label{lem:xlinkage}
There exists an algorithm which, given an instance of {\xalpp} on an
$n$-vertex graph $G$ with $r$ triples and maximum edge weight $W$, constructs
in time polynomial in $n+W$ an equivalent instance $(G',A,|A|/2,\ell)$ of {\falpp} with
the properties: 
(i) $|A|=2r$, (ii) $\pw(G')\le \pw(G)+2$.
\end{lemma}

\begin{proof}

First, we simplify the given instance of {\xalpp} by removing edge weights:
for every edge $e=\{u,v\}\in E(G)$ with $w(e)>1$, we remove this edge and replace
it with a path from $u$ to $v$ with length $w(e)$ going through new vertices
(in other words, $e$ has been subdivided $w(e)-1$ times). It is not hard to see that we
have an equivalent instance of {\xalpp} on the new graph, which we call $G_1$, where the
weight of all edges is $1$ and $|V(G_1)| \le n^2W$. 
We now give a polynomial-time reduction from this new
instance of {\xalpp} to {\falpp}.

Let $p=|V(G_1)|$ and $\ell = 2 p^{2}$. For each $i\in[r]$ we do the following: we construct a new
vertex $s_i'$ and connect it to $s_i$ using a path of length $p^{2} + ip$
going through new vertices; we construct a new vertex $t_i'$ and connect it to
$t_i$ using a path of length $p^{2}-ip -\ell_i$ through new vertices.
(Note that $p^{2}-ip -\ell_{i} > 0$ since $p \ge n \ge 2$, $i \le n/2$, and $\ell_{i} < n$.)
We set $A$ to contain all the vertices $s_i', t_i'$ for $i\in[r]$. This
completes the construction and it is clear that $|A|=2r$, the new graph $G'$ has order at most 
$2p^{3} \le 2 n^{6} \cdot W^{3}$ and
can be constructed in time polynomial in $n+W$. 

We claim that the new graph $G'$ has $|A|/2$ vertex-disjoint $(A,\ell)$-paths if and only if the
{\xalpp} instance of $G_1$ has a positive answer. Indeed, if there exists a
collection of $r$ vertex-disjoint paths in $G_1$ such that the $i$-th path has
endpoints $s_i,t_i$ and length $\ell_i$, we add to this path the paths from $s_i'$ to $s_i$ and from $t_i$ to $t_i'$
and this gives a path of length $\ell = 2p^{2}$ with
endpoints in $A$. Observe that all these paths are vertex-disjoint, so we
obtain a yes-certificate of {\falpp}. For the converse direction, suppose that 
$G'$ has a set $\mathcal{A}$ of $|A|/2$ vertex-disjoint $(A,\ell)$-paths. 
The set $\mathcal{A}$ does not contain a path with endpoints $s_{i}'$ and $s_{j}'$
since such a path has length at least $2p^{2} + (i+j) p + 1 > \ell$.
Furthermore, a path in $\mathcal{A}$ cannot connect some $s_{i}'$ and $t_{j}'$ with $i > j$
since the length of such a path is at least
$2p^{2} + (i-j)p-\ell_j +1 > \ell$.
Since $A = \{s_{i}' \mid i \in [r]\} \cup \{t_{i}' \mid i \in [r]\}$,
we can conclude that each path $P$ in $\mathcal{A}$ connects $s_{i}'$ and $t_{i}'$ for some $i$, and 
the subpath of $P$ connecting $s_{i}$ and $t_{i}$ has length
$2p^{2} - (p^{2} + ip) - (p^{2} - ip - \ell_{i}) = \ell_{i}$.
We therefore obtain a solution to the {\xalpp} instance.

Finally, observe that the only modifications we have done on $G$ is 
to subdivide some edges and to attach paths to some vertices.
By Corollary~\ref{cor:pw-subdivision-path}, the pathwidth is increased only by at most $2$.
\qed
\end{proof}

We can now reduce the $k$-\textsc{MCC} problem to {\xalpp}.

\begin{lemma}\label{lem:xlinkage2} There exists a polynomial-time algorithm
which, given an instance of $k$-\textsc{MCC} on a graph $G$ with $n$ vertices,
produces an equivalent instance of {\xalpp} on a graph $G'$, with $r \in O(k^2)$
triples, $\pw(G') \in O(k^2)$, and maximum edge weight $W \in n^{O(1)}$.  \end{lemma}

\begin{proof}
We are given a graph $G = (V,E)$ with $V$ partitioned into $k$ sets $V_{1}, \dots, V_{k}$,
and are asked for a clique of size $k$ that contains one vertex from each set. 
To ease notation, we will assume that $n$ is odd and $|V_{i}|=n$ for $i \in [k]$ 
(so the graph has $kn$ vertices in total) and that the vertices of $V_{i}$ are numbered $1, \dots, n$.
We define two lengths $L_{1} = (k+1)(n-1)$ and $L_{2} = 60 n^{6}$.

For $i\in[k]$ we construct a vertex-selection gadget as follows (see \figref{fig:vertex-gadget}):
we make $n$ paths of length $k$, call them $P_{i,j}$, where $j \in [n]$.
Let $a_{i,j}$ and $b_{i,j}$ be the first and last vertices of path $P_{i,j}$, respectively.
We label the remaining vertices of the path $P_{i,j}$ as $x_{i, j, i'}$ for
$i' \in \{1, \dots, k\} \setminus \{i\}$ in some arbitrary order.
Then for each $j \in [n-1]$ we connect $a_{i,j}$ to $a_{i,j+1}$ and $b_{i,j}$ to $b_{i,j+1}$.
All edges constructed so far have weight $1$.
We set $s_{i} = a_{i,1}$ and $t_{i} = a_{i,n}$.
We add to the instance the triple $(s_{i},t_{i},L_{1})$.

\begin{figure}[bth]
  \centering
  \includegraphics[scale=1]{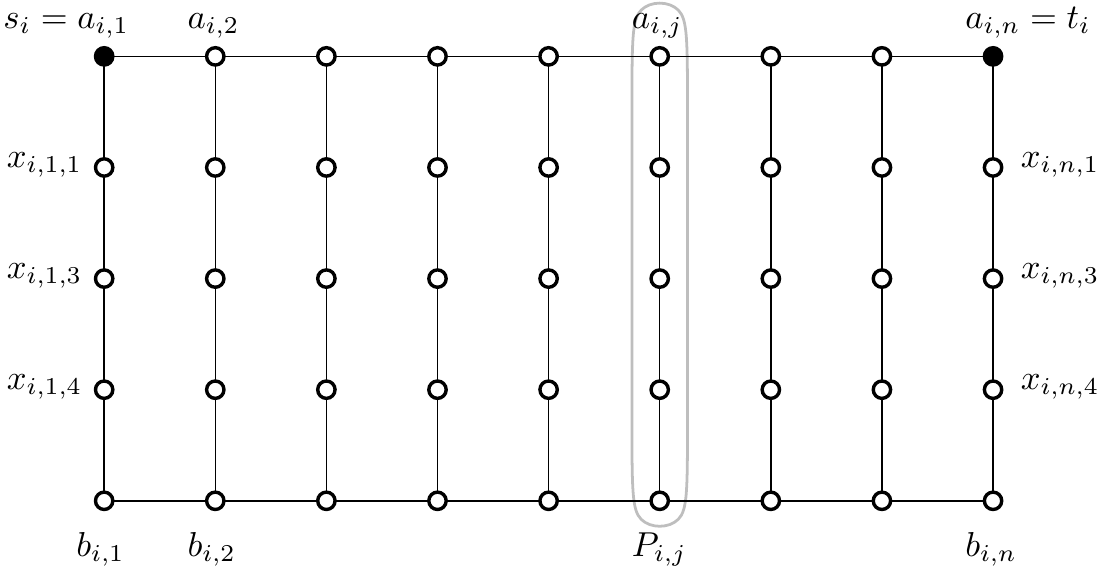}
  \caption{An example of the vertex-selection gadget for $n=9$, $k=4$, and $i=2$.}
  \label{fig:vertex-gadget}
\end{figure}

We now need to construct an edge-verification gadget as follows (see \figref{fig:edge-gadget}): 
for each $i_{1}, i_{2} \in [k]$ with $i_{1} < i_{2}$,
we construct three vertices $s_{i_1,i_2}$, $t_{i_1,i_2}$, $p_{i_1,i_2}$.
For each edge $e$ of $G$ between $V_{i_{1}}$ and $V_{i_{2}}$, we do the following: 
suppose $e$ connects vertex $j_{1}$ of $V_{i_{1}}$ to vertex $j_{2}$ of $V_{i_{2}}$.
We add the following four edges:
\begin{enumerate}

\item An edge from $s_{i_1,i_2}$ to $x_{i_1, j_1, i_2}$. This edge has weight
$L_{2}/4 + j_{1} n^{4} + j_{2} n^{2}$.

\item An edge from $x_{i_1,j_1,i_2}$ to $p_{i_1,i_2}$. This edge has weight
$L_2/4$.

\item An edge from $p_{i_1,i_2}$ to $x_{i_2,j_2,i_1}$. This edge has weight
$L_2/4$.

\item An edge from $x_{i_2,j_2,i_1}$ to $t_{i_1,i_2}$. This edge has weight
$L_{2}/4 - j_{1} n^{4} - j_{2} n^{2}$.

\end{enumerate}

We call the edges constructed in the above step heavy edges, since their weight
is close to $L_2/4$.  We add the $k(k-1)/2$ triples $(s_{i_1,i_2}, t_{i_1,i_2},
L_2)$ to the instance, for all $i_1,i_2\in[k]$, with $i_1< i_2$.

\begin{figure}[bth]
  \centering
  \includegraphics[scale=1]{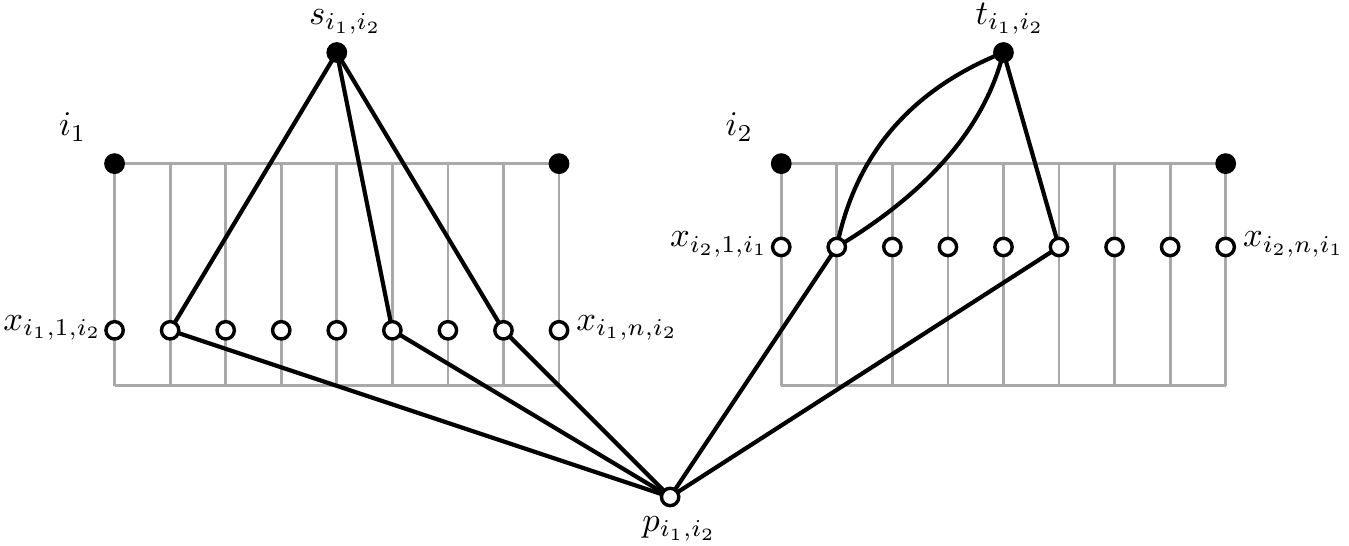}
  \caption{An example of the edge-verification gadget for $V_{i_{1}}$ and $V_{i_{2}}$ ($i_{1} < i_{2}$).
  In this example, there are exactly three edges between $V_{i_{1}}$ and $V_{i_{2}}$.}
  \label{fig:edge-gadget}
\end{figure}

Note that in the above description we have created some parallel edges, for
example from $s_{i_1,i_2}$ to $x_{i_1,2j_1,i_2}$ (if the vertex $j_1$ of
$V_{i_1}$ has several neighbors in $V_{i_2}$). This can be avoided by
subdividing such edges once and assigning weights to the new edges so that the
total weight stays the same. For simplicity we ignore this detail in the
remainder since it does not significantly affect the pathwidth of
the graph (see Corollary~\ref{cor:pw-subdivision-path}).  This completes the construction.

Let us now prove correctness. First assume that we have a $k$-multicolored
clique in $G$, encoded by a function $\sigma\colon [k] \to [n]$, that is, $\sigma(i)$
is the vertex of the clique that belongs in $V_i$. For the $i$-th vertex-selection 
gadget we have the triple $(s_i,t_i,L_1)$.
We construct a path from $s_i$ to $t_i$
by traversing the paths $P_{i,j}$ for $j \in [n] \setminus \{\sigma(i)\}$ 
in the increasing order of $j$
and by appropriately using the ``horizontal'' edges connecting adjacent paths.
See \figref{fig:correctness1}.
The path has length $L_{1}$:
we have traversed $n-1$ paths $P_{i,j}$ with $j \ne \sigma(i)$, each of which has $k$ edges;
we have also traversed $n-1$ horizontal edges connecting adjacent paths.
The total length is therefore, $(n-1)k + n-1 = L_{1}$. In this way we have satisfied all the
$k$ triples $(s_i,t_i,L_1)$ and have not used the vertices
$x_{i,\sigma(i),i'}$ for any $i'\neq i$.

Consider now a triple $(s_{i_1,i_2}, t_{i_1,i_2},L_2)$, for $i_1<i_2$. Because
we have selected a clique, there exists an edge between vertex $\sigma(i_1)$ of
$V_{i_1}$ and $\sigma(i_2)$ of $V_{i_2}$. For this edge we have constructed
four edges in our new instance, linking $s_{i_1,i_2}$ to $t_{i_1,i_2}$ with a
total weight of $L_2$. We use these paths to satisfy the $k\choose 2$ triples
$(s_{i_1,i_2},t_{i_1,i_2},L_2)$. These paths are disjoint from each other: when
$i_1<i_2$, $x_{i_1,\sigma(i_1),i_2}$ is only used in the path from $s_{i_1,i_2}$ to
$t_{i_1,i_2}$ and when $i_1>i_2$, $x_{i_1,\sigma(i_1),i_2}$ is only used in the
path from $s_{i_2,i_1}$ to $t_{i_2,i_1}$.  Furthermore, these paths are disjoint from
the paths in the vertex-selection gadgets, as we observed that
$x_{i,\sigma(i),i'}$ are not used by the path connecting $s_i$ to $t_i$. We
thus have a valid solution.
See \figref{fig:correctness1}.

\begin{figure}[bth]
  \centering
  \includegraphics[scale=1]{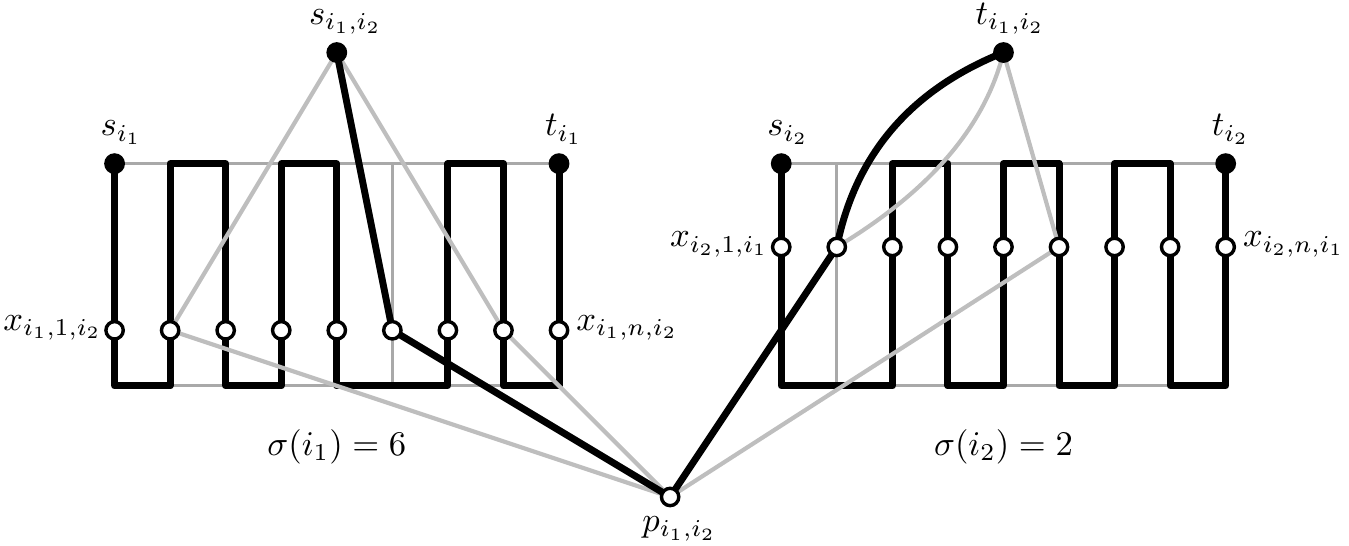}
  \caption{Construction of paths from $\sigma$.}
  \label{fig:correctness1}
\end{figure}

For the converse direction, suppose we have a valid solution for the {\xalpp}
instance. First, consider the path connecting $s_i$ to $t_i$. This path has
length $L_1$, therefore it cannot be using any heavy edges, since these edges
have cost at least $L_2/4-n^{5}-n^{3}>L_{1}$. Inside the vertex-selection gadget, the
path may use either all of the edges of a path $P_{i,j}$ or none. Let us now
see how many $P_{i,j}$ are unused. First, a simple parity argument shows that
the number of paths traversed in the $a_{i,j}\to b_{i,j}$ direction is equal to those
traversed in the opposite direction, so the total number of used paths is even.
Since we have an odd number of paths in total (as $n$ is odd), at least one path is not used.
We conclude that exactly one $P_{i,j}$ is not used, otherwise the path from
$s_i$ to $t_i$ would be too short. Let $\sigma(i)$ be defined as the index $j$
such that the internal vertices of $P_{i,j}$ are not used in the $s_i\to t_i$
path of the solution. We define a clique in $G$ by selecting for each $i$ the
vertex $\sigma(i)$.

Let us argue why this set induces a clique. Let $j_1,j_2$ be the vertices
selected in $V_{i_1},V_{i_2}$ respectively, with $i_1<i_2$, and consider the
triple $(s_{i_1,i_2},t_{i_1,i_2},L_2)$. 
This triple must be satisfied by a path
that uses exactly four heavy edges, since each heavy edge has weight
at least $L_{2}/5 + n^{6}$ and at most $L_{2}/3 - 3n^{6}$,
and all other edges together have weight smaller than $n^{3}$.
Hence, every such path is using at least two internal vertices of
some $P_{i,j}$ because every heavy edge is incident on such a vertex.
By our previous reasoning, the paths that satisfy the $(s_i,t_i,L_1)$ triples have
used all such vertices except for one path $P_{i,j}$ for each $i$. There exist
therefore exactly $k(k-1)$ such vertices available, so each of the $k(k-1)/2$
triples $(s_{i_1,i_2}, t_{i_1,i_2}, L_2)$ has a path using exactly two of these
vertices. Hence, each such path consists of four heavy edges and no other edges.

Such a path must therefore be using one edge incident on $s_{i_1,i_2}$, one
edge incident on $t_{i_1,i_2}$ and two edges incident on $p_{i_1,i_2}$. The
used edge incident on $s_{i_1,i_2}$ must have as other endpoint
$x_{i_1,2j_1,i_2}$, which implies that its weight is $L_2/4+j_{1}n^{4}+j_{2}'n^{2}$,
for some $j_2'$.  Similarly, the edge incident on $t_{i_1,i_2}$ must have
weight $L_2/4-j_{1}'n^{4}-j_{2}n^{2}$, as its other endpoint is necessarily
$x_{i_2,2j_2,i_1}$. We conclude that the only way that the length of this path
is $L_2$ is if $j_1=j_1'$ and $j_2=j_2'$. Therefore, we have an edge between
the two selected vertices, and as a result a $k$-clique.

To conclude we observe that deleting the $3 \cdot \binom{k}{2}$
vertices $s_{i_1,i_2}$, $p_{i_1,i_2}$, $t_{i_1,i_2}$ disconnects the graph into components that correspond
to the vertex gadgets with some paths attached.
By Corollary~\ref{cor:pw-subdivision-path}, each such component has pathwidth at most 4 as it can be seen 
as a graph obtained from a subdivision of the $2 \times n$ grid by attaching paths to some vertices.
As a result the whole graph has pathwidth $3 \cdot \binom{k}{2} + 4$.
\qed
\end{proof}

\begin{theorem}
\label{thm:pw+|A|}
{\falpp} is W[1]-hard parameterized by $\pw+|A|$.
\end{theorem}

\begin{proof} We compose the reductions of Lemmas \ref{lem:xlinkage} and
\ref{lem:xlinkage2}.  Starting with an instance of $k$-\textsc{MCC} with $n$
vertices this gives an instance of {\falpp} with $n^{O(1)}$ vertices,
$|A|=O(k^2)$, and pathwidth $O(k^2)$. \qed \end{proof}


\section{Hardness on grid graphs}
\label{sec:grid-graphs}

In this section, we show that for every constant $\ell \ge 4$, {\falpp} is NP-complete on grid graphs.
We first reduce {\pcsat} to {\falpp} on planar bipartite graphs of maximum degree at most 4.
We then modify the instance by subdividing edges and adding terminal vertices in an appropriate way,
and have an equivalent instance on grid graphs.

The input of \textsc{Circuit SAT} is a Boolean circuit with a number of inputs and one output.
The question is whether the circuit can output \textsf{true} by appropriately setting its inputs.
\textsc{Circuit SAT} is NP-complete since \textsc{CNF SAT}~\cite{Cook71} can be seen as a special case. 
When the underlying graph of the circuit is planar, the problem is called {\pcsat}.
Using planar crossover gadgets~\cite{McColl81}, we can show that {\pcsat} is NP-complete.
Furthermore, since NOR gates can replace other gates such as AND, OR, NOT, NAND, and XOR
without introducing any new crossing, we can conclude that 
{\pcsat} having NOR gates only is NP-complete.

Let $I = (G, A, \ell)$ be an instance of {\falpp} with $G = (V,E)$.
Let $\psi$ be a mapping that assigns each $e \in E$ an $(A,\ell)$-path in $G$,
and $\psi(E) = \{\psi(e) \mid e \in E\}$.
We say that $\psi$ is a \emph{guide} to $I$
if every set of $|A|/2$ vertex-disjoint $(A,\ell)$-paths, if any exists, is a subset of $\psi(E)$.
When a guide is given additionally to an instance of {\falpp},
we call the problem \textsf{Guided} {\falpp}.
Observe that a guide to an instance is not a restriction
but just additional information.

\begin{lemma}
\label{lem:planar-bigraph}
For every fixed $\ell \ge 4$, \textsf{Guided} {\falpp} is NP-complete on planar bipartite graphs of maximum degree at most 4.
\end{lemma}
\begin{proof}
Given a planar circuit with only NOR gates, 
we construct an equivalent instance of \textsf{Guided} {\falpp} with the fixed $\ell$.
We only need input gadget, output gadget, split gadget, NOR gadget,
and a way to connect the gadgets.
See \figref{fig:circuit-to-alpp} for the high-level idea of the reduction.
\begin{figure}[hbt]
  \centering
  {}
  \hfill
  \includegraphics[scale=.7,align=c]{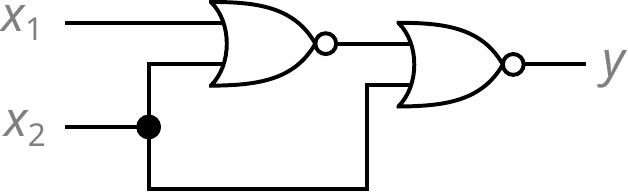}
  \hfill
  \includegraphics[scale=.7,align=c]{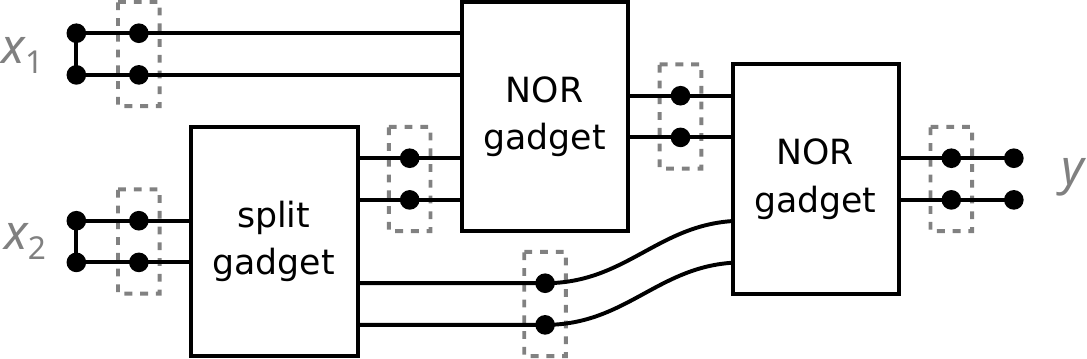}
  \hfill
  {}
  \caption{A planar circuit and the corresponding {\alpp} instance (simplified). 
    The vertices in $V \setminus A$ are omitted. The connection pairs are marked with dashed rectangles.}
  \label{fig:circuit-to-alpp}
\end{figure}

We explicitly present the gadgets for the cases $\ell = 4$ and $\ell = 5$.
For even (resp.\ odd) $\ell > 5$, the gadgets can be obtained from the one for $\ell = 4$ (resp.\ $\ell = 5$)
by subdividing $\lfloor \ell/2 \rfloor - 2$ times each edge incident to a vertex in $A$.

\paragraph{Connections between gadgets.}
We first explain how the gadgets are connected.
Each gadget has one or three pairs of vertices that are shared with other gadgets.
We call them \emph{connection pairs}.
All those vertices belong to the terminal set $A$.
In the figures, we draw each connection pair so that the two vertices are  next to each other vertically 
and mark them with a dashed rectangle.
If the $(A,\ell)$-paths using the vertices of a connection pair are going to the positive direction,
then we interpret it as that a \textsf{true} signal is sent via the connection pair.
If the paths are going to the negative direction,
then the connection pair is carrying a \textsf{false} signal.
(See \figref{fig:connection-pair}.)
Note that our reduction below forces the paths at each connection pair to proceed in the same direction.

\begin{figure}[hbt]
  \def\SCALE{.7}
  \centering
  {}
  \hfill
  \includegraphics[scale=\SCALE,align=c]{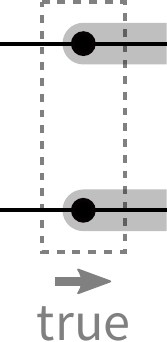}
  \hfill
  \includegraphics[scale=\SCALE,align=c]{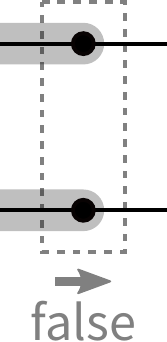}
  \hfill
  {}
  \caption{$(A,\ell)$-paths at a connection pair. We draw an $(A,\ell)$-path as a gray bar.}
  \label{fig:connection-pair}
\end{figure}

\paragraph{Input gadgets.}
The input gadget is simply a path of length $3\ell$, where the endpoints form its unique connection pair.
See \figref{fig:input_gadget}. For a full $(A,\ell)$-path packing, we only have two options.
One corresponds to \textsf{true} input~(\figref{fig:input_t})
and the other to \textsf{false} input~(\figref{fig:input_f}).
\begin{figure}[hbt]
  \def\SCALE{0.7}
  \def\MPSIZE{0.23}
  \centering
  \begin{minipage}[b]{\MPSIZE\linewidth}
    \centering
    \includegraphics[scale=\SCALE]{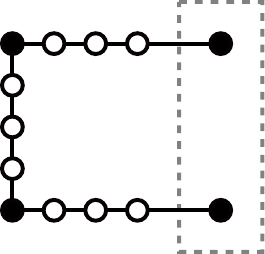}
    \subcaption{$\ell = 4$.}
    \label{fig:in_even}
  \end{minipage}
  \begin{minipage}[b]{\MPSIZE\linewidth}
    \centering
    \includegraphics[scale=\SCALE]{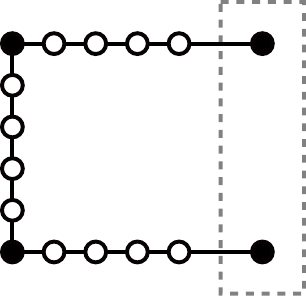}
    \subcaption{$\ell = 5$.}
    \label{fig:in_odd}
  \end{minipage}
  \begin{minipage}[b]{\MPSIZE\linewidth}
    \centering
    \includegraphics[scale=\SCALE]{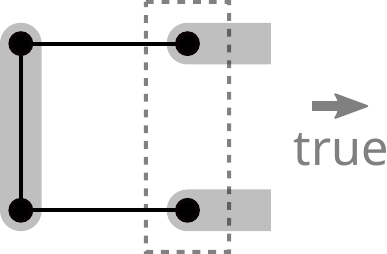}
    \subcaption{When input is \textsf{true}.}
    \label{fig:input_t}
  \end{minipage}
  \begin{minipage}[b]{\MPSIZE\linewidth}
    \centering
    \includegraphics[scale=\SCALE]{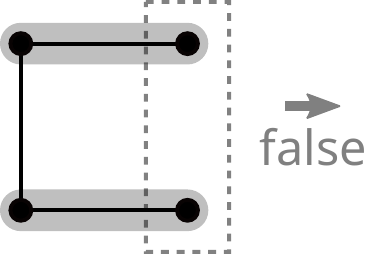}
    \subcaption{When input is \textsf{false}.}
    \label{fig:input_f}
  \end{minipage}
  \caption{The input gadgets. The black vertices belong to $A$ and the white vertices belong to $V \setminus A$.}
  \label{fig:input_gadget}
\end{figure}

\paragraph{Output gadgets.}
The output gadget consists of two paths of length $\ell$, where its unique connection pair includes one endpoint from each path.
See \figref{fig:output_gadget}. To have a full packing, the input to this gadget has to be \textsf{true}.
\begin{figure}[hbt]
  \def\SCALE{0.7}
  \def\MPSIZE{0.23}
  \centering
  \begin{minipage}[b]{\MPSIZE\linewidth}
    \centering
    \includegraphics[scale=\SCALE]{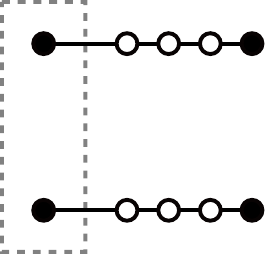}
    \subcaption{$\ell = 4$}
    \label{fig:out_even}
  \end{minipage}
  \begin{minipage}[b]{\MPSIZE\linewidth}
    \centering
    \includegraphics[scale=\SCALE]{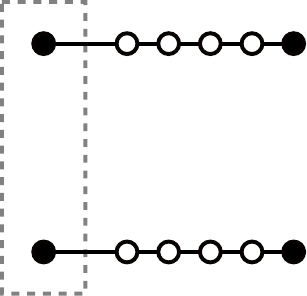}
    \subcaption{$\ell = 5$}
    \label{fig:out_odd}
  \end{minipage}
  \begin{minipage}[b]{\MPSIZE\linewidth}
    \centering
    \includegraphics[scale=\SCALE]{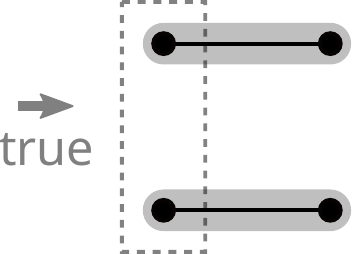}
    \subcaption{When output is \textsf{true}.}
    \label{fig:output_t}
  \end{minipage}
  \begin{minipage}[b]{\MPSIZE\linewidth}
    \centering
    \includegraphics[scale=\SCALE]{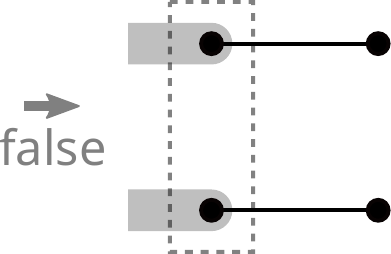}
    \subcaption{When output is \textsf{false}.}
    \label{fig:output_f}
  \end{minipage}
  \caption{The output gadgets.}
  \label{fig:output_gadget}
\end{figure}

\paragraph{Split gadgets.}
To simulate the split of a wire depicted in \figref{fig:split_circuit},
the split gadget consists of three paths of length $3\ell$, each of which is identical to the input gadget,
and a cycle of length $10\ell$ that synchronizes the three paths. See \figref{fig:split_gadget}.
To have a full $(A,\ell)$-path packing, there are only two ways to pack $(A,\ell)$-paths into a split gadget.
\figref{fig:split_gadget_packing} shows the two ways:
one on the left corresponds to a split of a \textsf{true} signal,
and the other a split of a \textsf{false} signal.

\begin{figure}[hbt]
  \centering
  \includegraphics[scale=0.7]{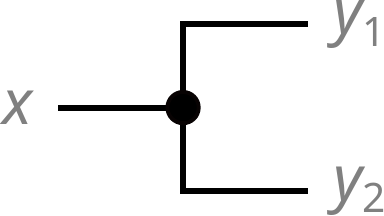}
  \caption{Splitting a wire in a circuit.}
  \label{fig:split_circuit}
\end{figure}

\begin{figure}[hbt]
  \def\SCALE{0.7}
  \def\MPSIZE{0.45}
  \centering
  \begin{minipage}[b]{\MPSIZE\linewidth}
    \centering
    \includegraphics[scale=\SCALE]{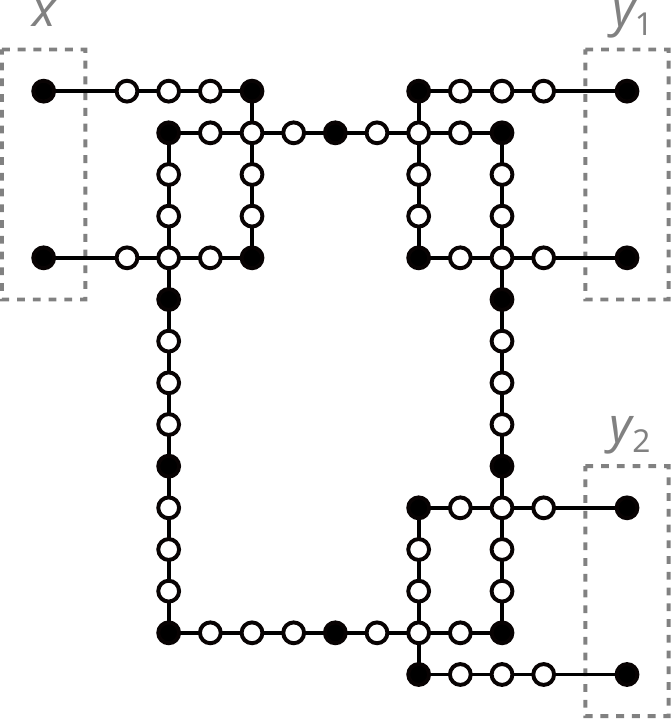}
    \subcaption{$\ell = 4$.}
    \label{fig:split_even}
  \end{minipage}
  \begin{minipage}[b]{\MPSIZE\linewidth}
    \centering
    \includegraphics[scale=\SCALE]{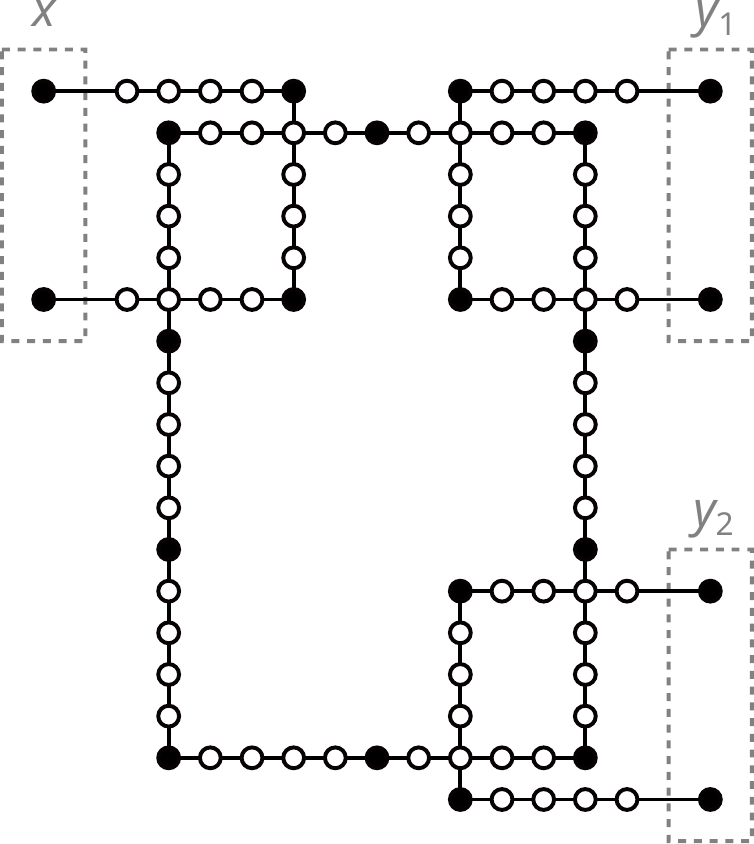}
    \subcaption{$\ell = 5$.}
    \label{fig:split_odd}
  \end{minipage}
  \caption{The split gadgets.}
  \label{fig:split_gadget}
\end{figure}

\begin{figure}[hbt]
  \def\SCALE{0.7}
  \def\MPSIZE{0.45}
  \centering
  \begin{minipage}[b]{\MPSIZE\linewidth}
    \centering
    \includegraphics[scale=\SCALE]{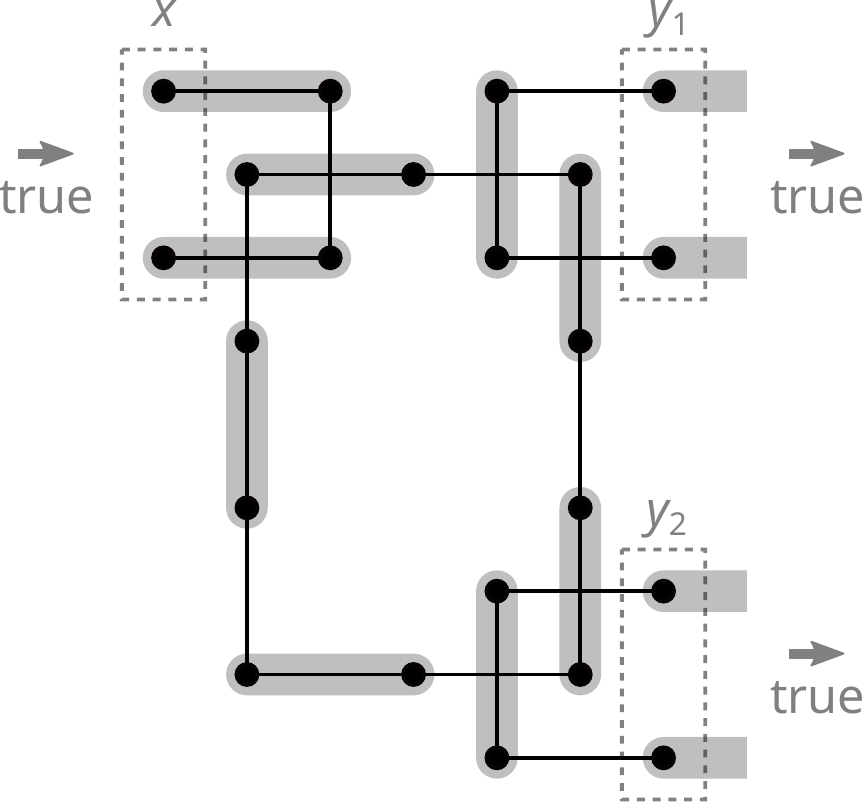}
    \subcaption{Splitting \textsf{true} signal.}
    \label{fig:split_t}
  \end{minipage}
  \begin{minipage}[b]{\MPSIZE\linewidth}
    \centering
    \includegraphics[scale=\SCALE]{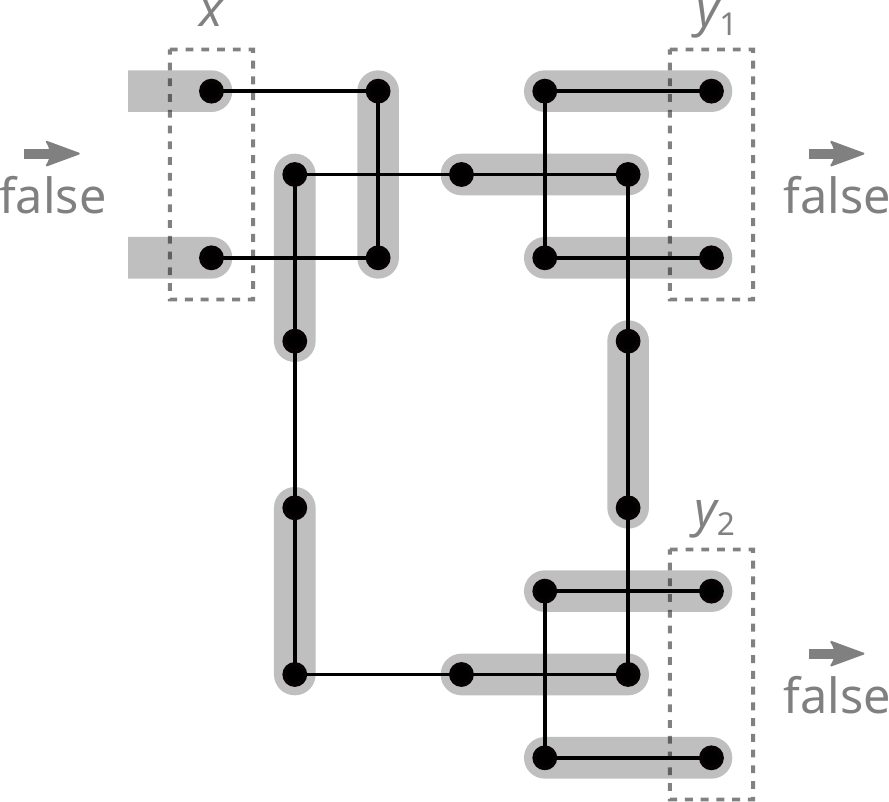}
    \subcaption{Splitting \textsf{false} signal.}
    \label{fig:split_f}
  \end{minipage}
  \caption{The possible $(A,\ell)$-path packings of the split gadget.}
  \label{fig:split_gadget_packing}
\end{figure}

\paragraph{NOR gadgets.}
Recall that NOR stands for ``NOT OR''
and that the output $y$ of a NOR gate is \textsf{true} if and only if both inputs $x_{1}$ and $x_{2}$ are \textsf{false}.
The NOR gadgets are given in \figref{fig:nor_gadget}.
The structure of the gadget is rather involved.
It has three connection pairs, two for the inputs and one for the output, and the endpoints of each pair are connected by a path of length $5\ell$.
Additionally, there is a long self-intersecting cycle that somehow entangles the inputs and the output.
There are only four ways to fully pack $(A,\ell)$-paths into a NOR gadget,
and each packing corresponds to a correct behavior of a NOR gate (see \figref{fig:nor_gadget_packing}).
To see the correctness of \figref{fig:nor_gadget_packing},
it is important to observe that 
in the NOR gadgets for even $\ell$, 
there are some $(A,\ell)$-paths that are never used in a full $(A,\ell)$-path packing.
For example, the $(A,\ell)$-path with endpoints $v_{1}$ and $v_{2}$ in \figref{fig:nor_even} is such a path.
In a full $(A,\ell)$-path packing, $w_{2}$ has to be an endpoint of an $(A,\ell)$-path either with $w_{1}$ or $w_{3}$.
Hence, if we use the $(A,\ell)$-path with endpoints $v_{1}$ and $v_{2}$,
then one of $u_{1}$ and $u_{2}$ cannot belong to any $(A,\ell)$-path in the packing.
\begin{figure}[hbt]
  \def\SCALE{.55}
  \centering
  \begin{minipage}[b]{.48\linewidth}
    \centering
    \includegraphics[scale=\SCALE]{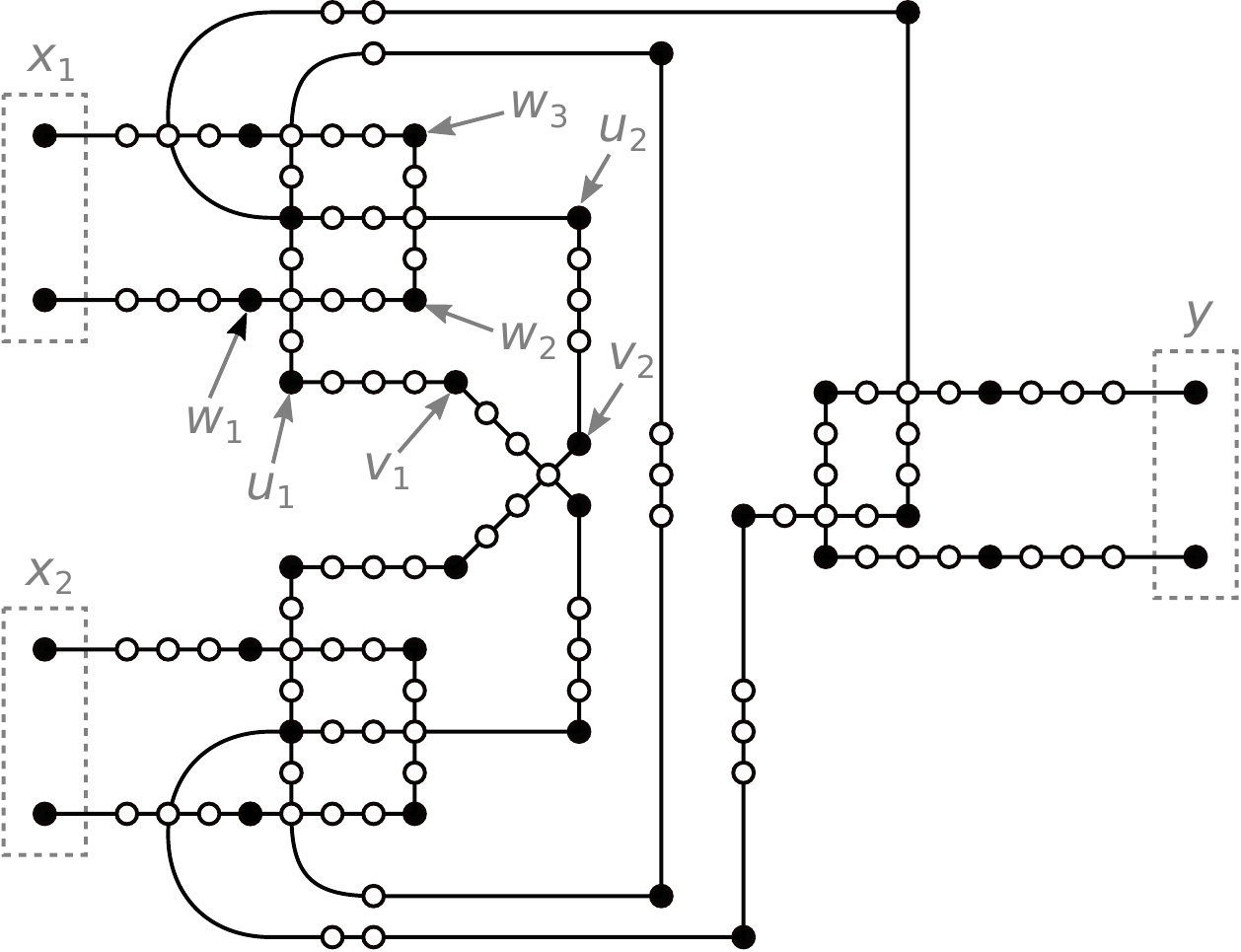}
    \subcaption{$\ell = 4$.}
    \label{fig:nor_even}
  \end{minipage}
  \begin{minipage}[b]{.49\linewidth}
    \centering
    \includegraphics[scale=\SCALE]{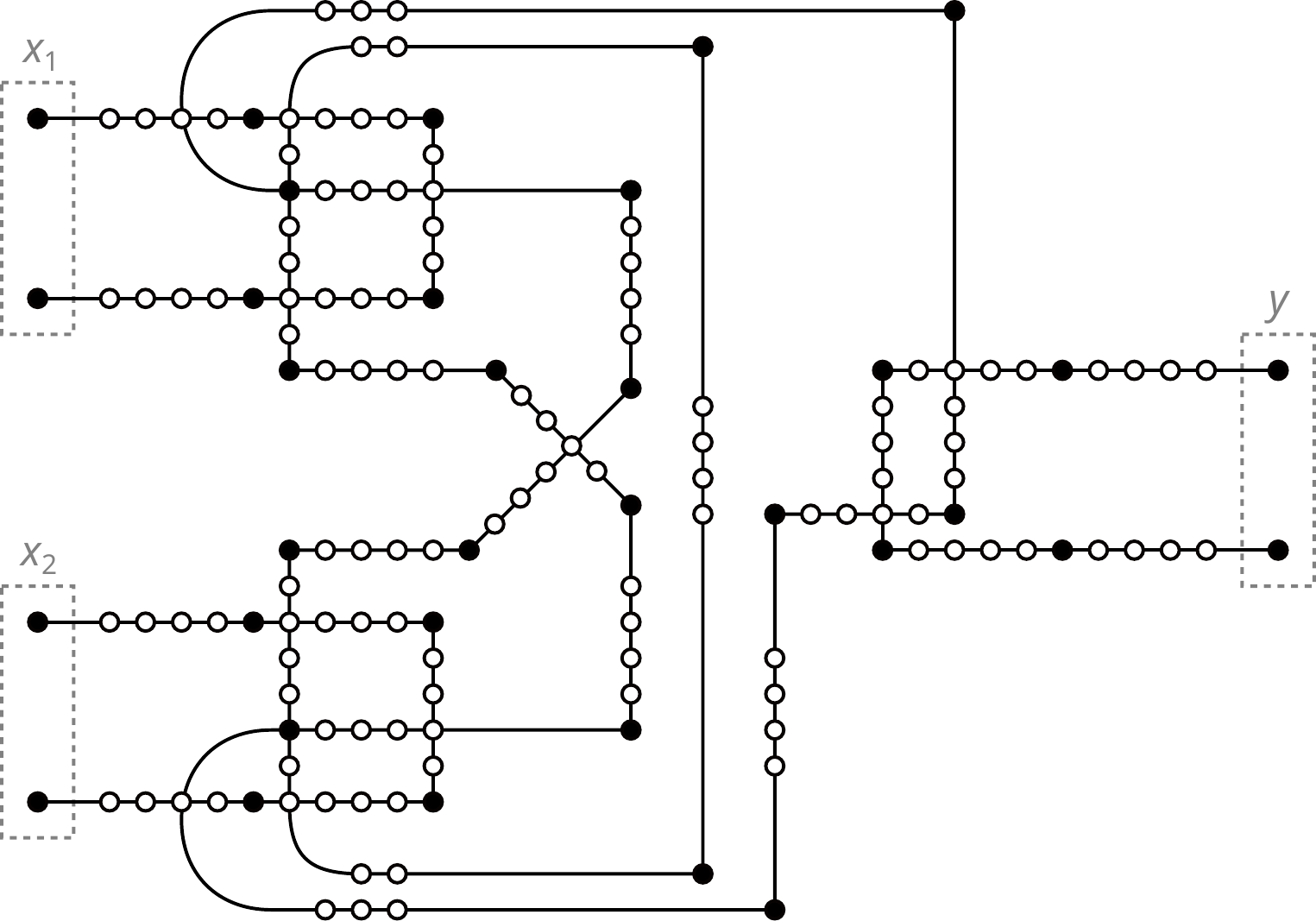}
    \subcaption{$\ell = 5$.}
    \label{fig:nor_odd}
  \end{minipage}
  \caption{The NOR gadgets.}
  \label{fig:nor_gadget}
\end{figure}

\begin{figure}[hbt]
  \def\SCALE{0.5}
  \def\MPSIZE{0.48}
  \centering
  \begin{minipage}[b]{\MPSIZE\linewidth}
    \centering
    \includegraphics[scale=\SCALE]{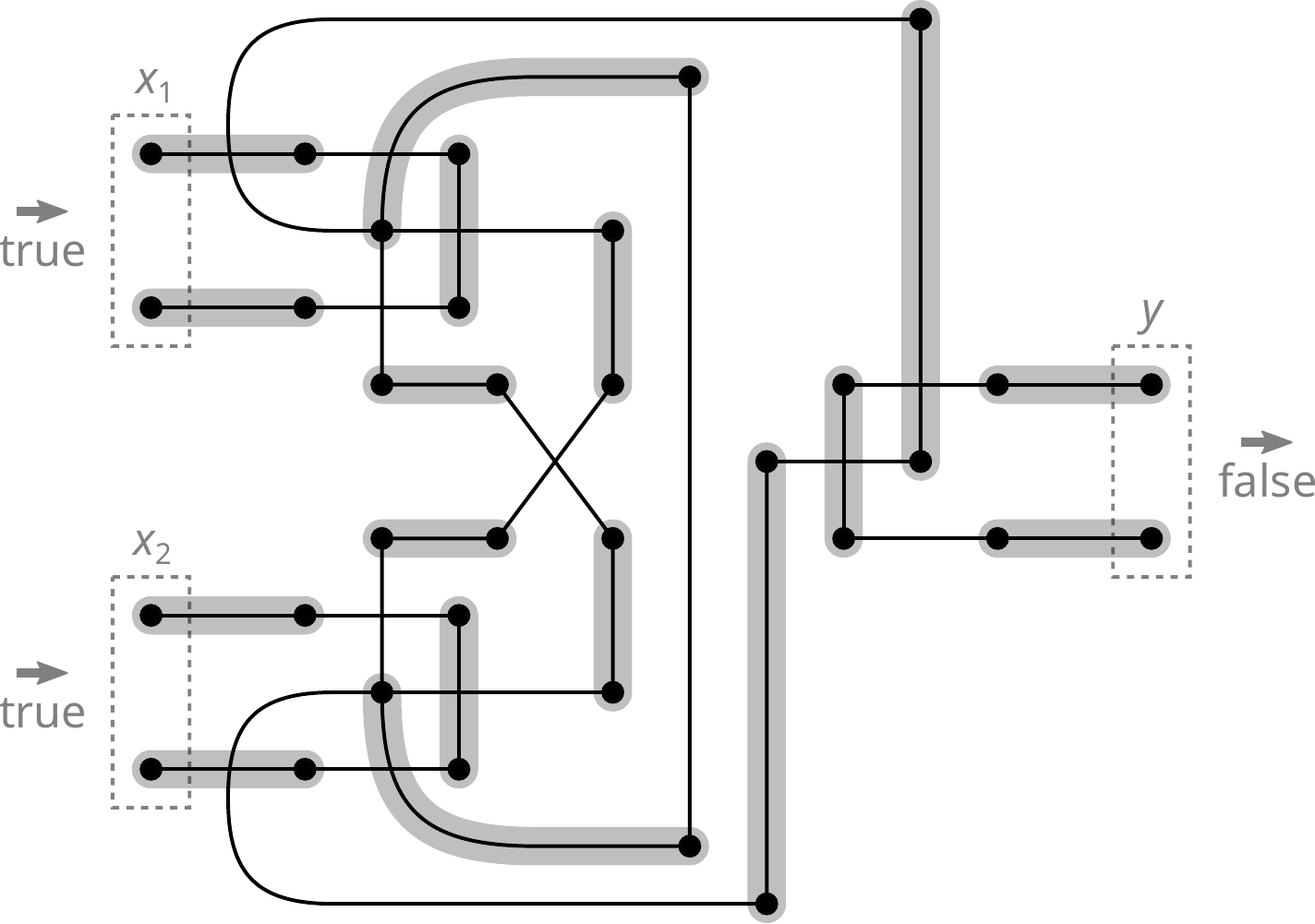}
    \subcaption{When input is (\textsf{true, true}).}
    \label{fig:nor_tt}
  \end{minipage}
  \begin{minipage}[b]{\MPSIZE\linewidth}
    \centering
    \includegraphics[scale=\SCALE]{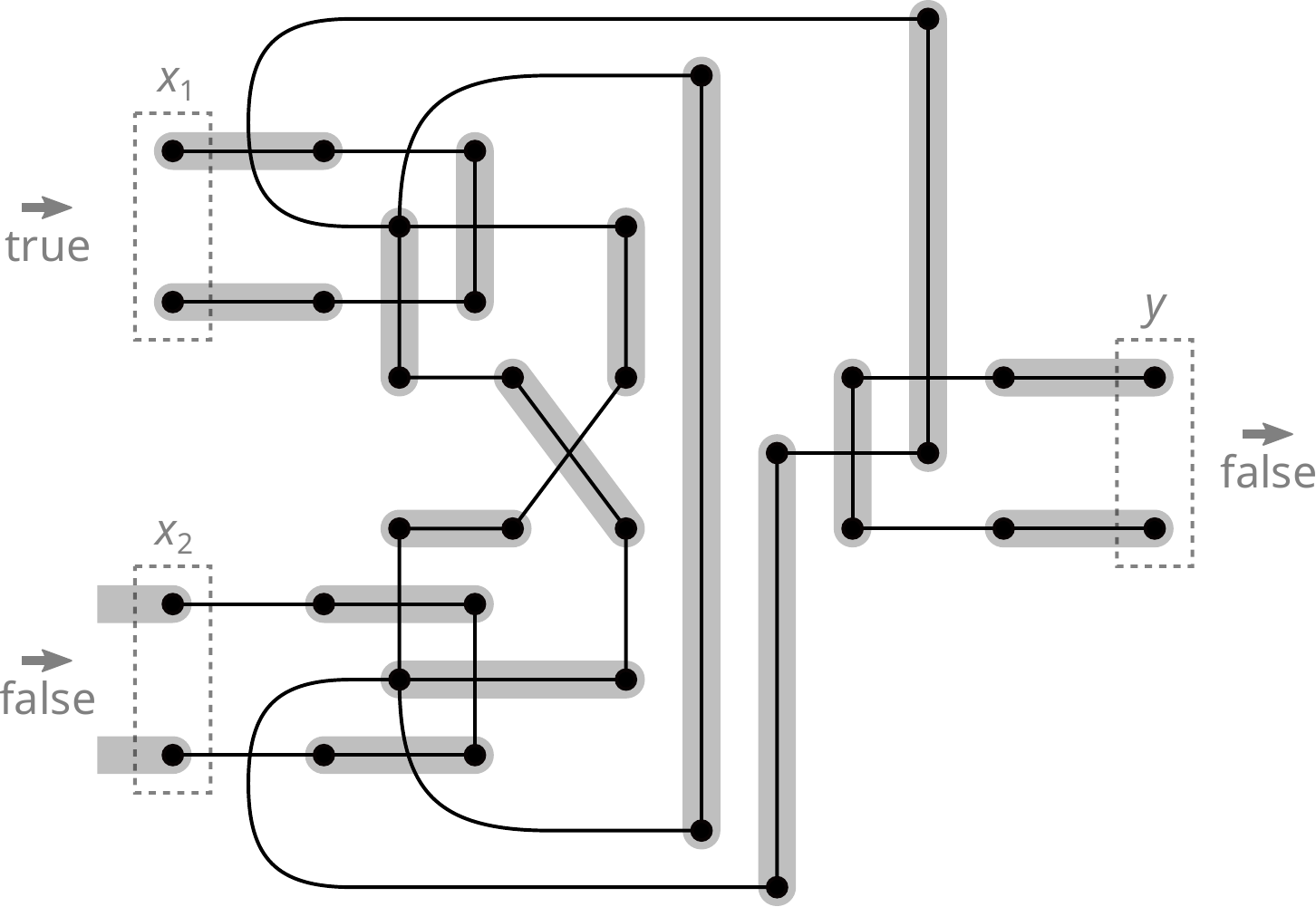}
    \subcaption{When input is (\textsf{true, false}).}
    \label{fig:nor_tf}
  \end{minipage}
  
  \bigskip
  \begin{minipage}[b]{\MPSIZE\linewidth}
    \centering
    \includegraphics[scale=\SCALE]{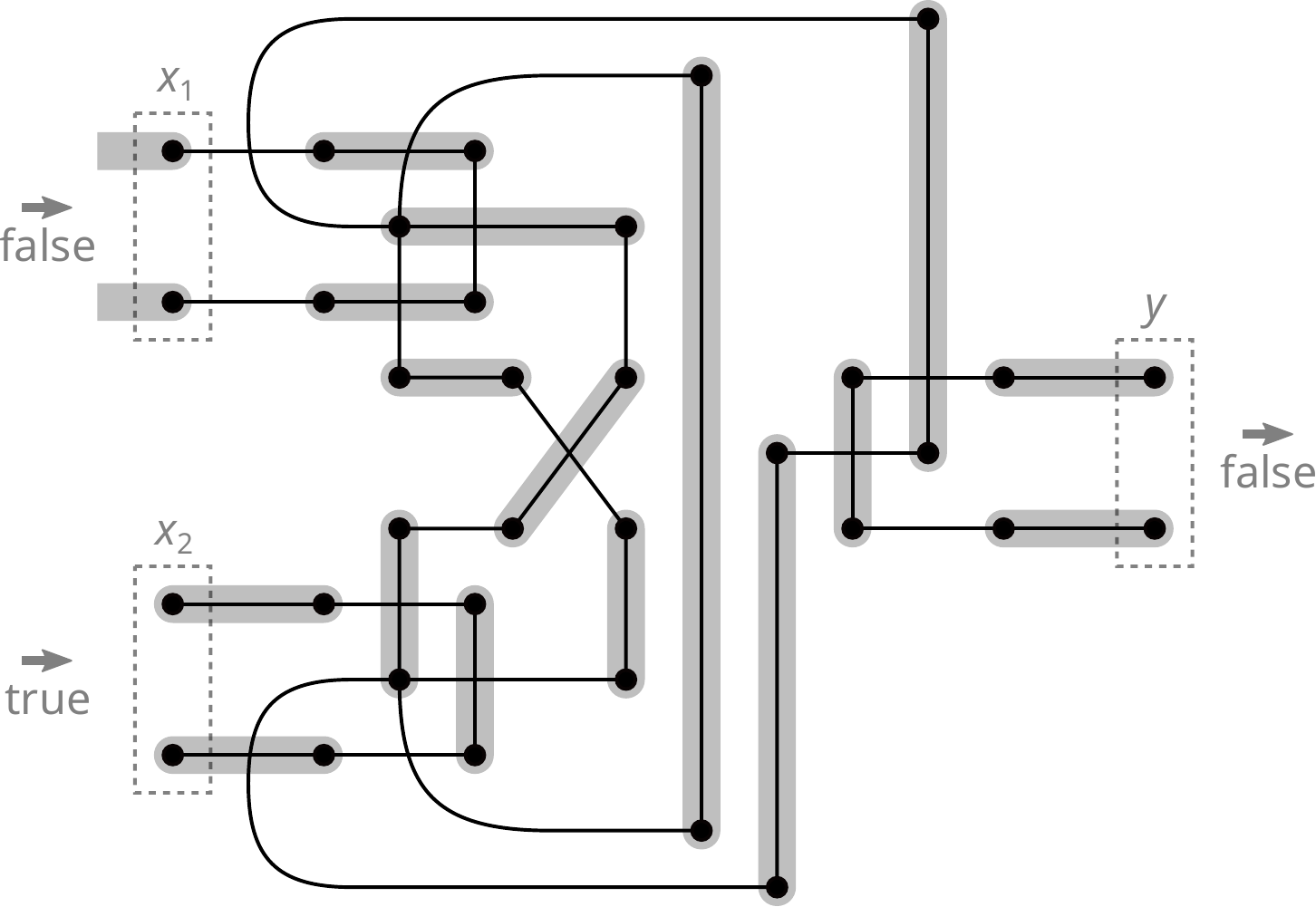}
    \subcaption{When input is (\textsf{false, true}).}
    \label{fig:nor_ft}
  \end{minipage}
  \begin{minipage}[b]{\MPSIZE\linewidth}
    \centering
    \includegraphics[scale=\SCALE]{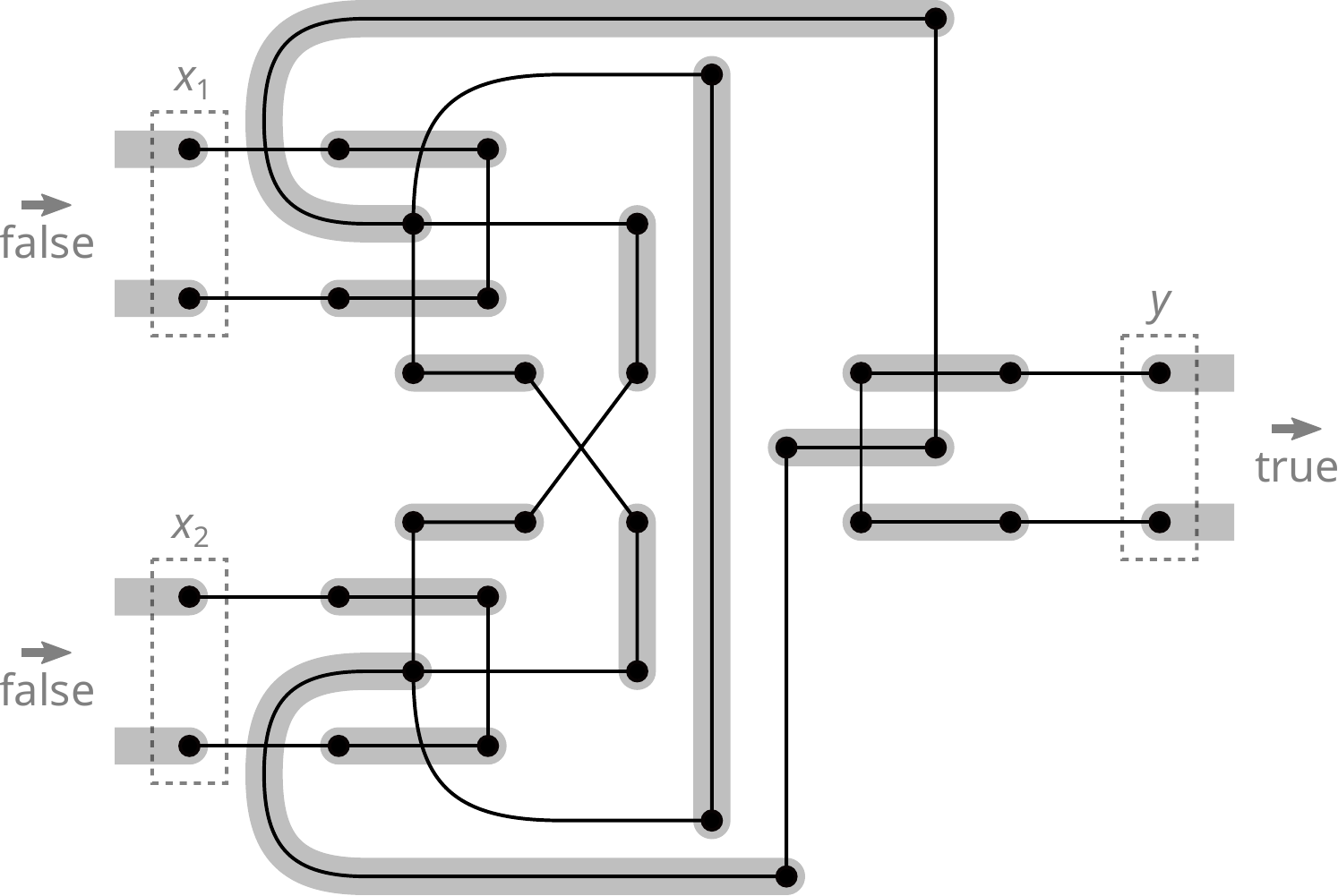}
    \subcaption{When input is (\textsf{false, false}).}
    \label{fig:nor_ff}
  \end{minipage}
	\caption{The possible $(A,\ell)$-path packings of the NOR gadget.}
  \label{fig:nor_gadget_packing}
\end{figure}

\paragraph{Guides.}
The guide $\psi(e)$ for each $e \in E$ can be easily set from 
Figures~\ref{fig:input_t}, \ref{fig:input_f}, \ref{fig:output_t},
\ref{fig:split_t}, \ref{fig:split_f},
\ref{fig:nor_tt}, \ref{fig:nor_tf}, \ref{fig:nor_ft}, and \ref{fig:nor_ff}.
For each edge $e$, the unique gray bar that includes the edge represents the $(A,\ell)$-path $\psi(e)$.

\paragraph{Correctness.}
The correctness of each gadget implies the correctness of the whole reduction.
Thus, it suffices to show that the output of the reduction is planar bipartite graph of maximum degree at most 4.
The resultant graph clearly has maximum degree 4 and is planar.
To see that the graph is bipartite, consider a 2-coloring of a gadget, which is not the output gadget.
If $\ell$ is even, then all vertices in the connection pairs have the same color.
If $\ell$ is odd, then each upper vertex of a connection pair in the figures has the same color,
and the other vertices in the connection pairs have the other color.
Therefore, the entire graph is 2-colorable.
\qed
\end{proof}

Let $(G,A,\ell,\psi)$ be an instance of \textsf{Guided} {\falpp} with $G = (V,E)$.
For $e = \{v,w\} \in E$,
we denote by $d_{e}$ the length of the subpath of $\psi(e)$ starting at $v$, passing $w$, and reaching an endpoint of $\psi(e)$.
Let $G_{e}$ be the graph obtained from $G$ by subdividing $e$, $2\ell$ times.
Let $P_{e}$ be the $v$-$w$ path of length $2\ell+1$ in $G_{e}$ corresponding to $e$.
We set $A_{e} = A \cup \{x_{0}, x_{1}\}$,
where $x_{0}$ and $x_{1}$ are the vertices that have distance $d_{e}$ and $d_{e} + \ell$ from $v$ in $P_{e}$, respectively.
(See \figref{fig:subdiv}.)
\begin{figure}[hbt]
  \centering
  \includegraphics[width=.9\linewidth]{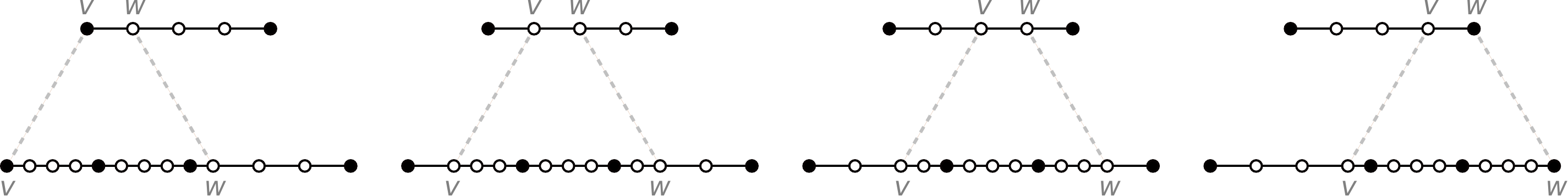}
  \caption{Subdividing $e$ and introducing two new terminals ($\ell = 4$).}
  \label{fig:subdiv}
\end{figure}

For each edge $h$ of $G_{e}$, we set $\psi_{e}(h) = \psi(h)$ if $h$ is not contained in the path $P_{\psi(e)}$ of length $3\ell$
that corresponds to $\psi(e)$. 
If $h$ is contained in $P_{\psi(e)}$, then we set $\psi_{e}(h)$ to the unique $(A,\ell)$-path in $P_{\psi(e)}$ that contains $h$.
Observe that $\psi_{e}$ is a guide to $(G_{e}, A_{e}, \ell)$.
Furthermore, $(G,A,\ell,\psi)$ and $(G_{e},A_{e},\ell, \psi_{e})$ are equivalent
(see \figref{fig:path_equiv}):
if $\psi(e)$ is used in a full $(A,\ell)$-path packing of $G$,
then we use two $(A,\ell)$-paths in $P_{\psi(e)}$;
otherwise we use the middle $(A,\ell)$-path in $P_{\psi(e)}$ connecting two new terminals.
\begin{observation}
\label{obs:subdiv}
$(G,A,\ell,\psi)$ and $(G_{e},A_{e},\ell, \psi_{e})$ are equivalent instances of \textsf{Guided} {\falpp}.
\end{observation}
\begin{figure}[hbt]
  \centering
  \includegraphics[scale=.7]{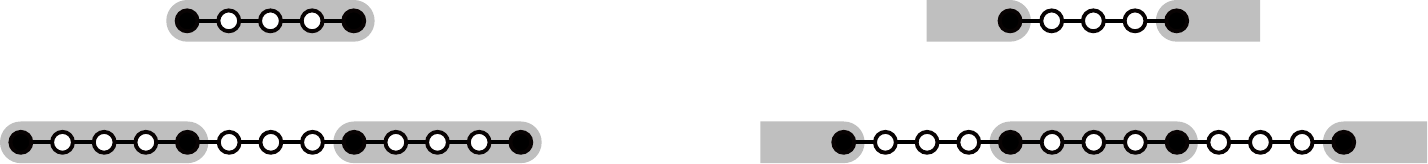}
  \caption{Equivalence of $(G,A,\ell,\psi)$ and $(G_{e},A_{e},\ell, \psi_{e})$.}
  \label{fig:path_equiv}
\end{figure}

Now we are ready to prove the main theorem of this section.
\begin{theorem}
\label{thm:grid}
For every constant $\ell \ge 4$, {\falpp} is NP-complete on grid graphs.
\end{theorem}
\begin{proof}
We reduce \textsf{Guided} {\falpp} on planar bipartite graphs of maximum degree at most 4
for fixed $\ell \ge 4$ (which is NP-complete by Lemma~\ref{lem:planar-bigraph})
to {\falpp} on grid graphs for the same $\ell$.
Let $(G,A,\ell,\psi)$ be an instance of \textsf{Guided} {\falpp},
where $G = (V,E)$ is a planar bipartite graph of maximum degree at most 4.

A \emph{rectilinear embedding} of a graph is a planar embedding into the $\mathbb{Z}^{2}$ grid such that
\begin{itemize}
  \item each vertex is mapped to a grid point;
  \item each edge $\{u,v\}$ is mapped to a rectilinear path between $u$ and $v$ consisting of vertical and horizontal segments connecting grid points;
  \item the rectilinear paths corresponding to two different edges may intersect only at their endpoints.
\end{itemize}
Every planar graph of maximum degree at most 4 has
a rectilinear embedding,
and a rectilinear embedding of area at most $(n+1)^{2}$ can be computed in linear time~\cite{Liu1998}, where $n$ is the number of vertices.

Let $R_{1}$ be a rectilinear embedding of $G$ with area at most $(n+1)^{2}$.
By multiplying each coordinate in the embedding by $2\ell$,
we obtain an enlarged rectilinear embedding $R_{2}$ of $G$.
Let $U$ be one color class of a 2-coloring of $G$.
Now, for each $v \in U$, we locally modify $R_{2}$ around the grid point $(x_{v}, y_{v})$ corresponding to $v$ as illustrated in \figref{fig:deform}.
We denote by  $R_{3}$ the locally modified embedding.

\begin{figure}[hbt]
  \centering
  \includegraphics[scale=1]{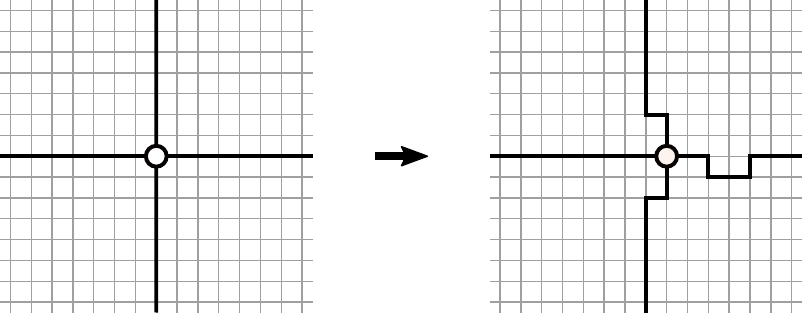}
  \caption{Local modification around $v$. The grid point of $v$ is moved to $(x_{v} + 1, y_{v})$.}
  \label{fig:deform}
\end{figure}

From $R_{3}$, we construct a new graph $G'$ and its rectilinear embedding $R'$ by inserting degree-2 vertices
at each intersection point of a grid point and the inner part of a rectilinear path corresponding to an edge.
Clearly, $G'$ is a grid graph.
Let $e \in E$ and $\lambda_{e}$ be the (geometric) length of the rectilinear path in $R_{1}$ corresponding to $e$.
Then the rectilinear path in $R_{3}$ corresponding to $e$ has length $2 \ell \cdot \lambda_{e} + 1$.
Therefore, $G'$ is the graph obtained from $G$ by subdividing each edge $e$, $2 \ell \cdot \lambda_{e}$ times.
By Observation~\ref{obs:subdiv}, we can easily compute $A'$ and $\psi'$
such that $(G,A,\ell,\psi)$ is equivalent to $(G',A',\ell,\psi')$.
Finally, from the definition of a guide to an instance of {\falpp},
$(G',A',\ell,\psi')$ is equivalent to $(G',A',\ell)$.
As everything in this reduction can be done in time polynomial, the theorem holds.
\qed
\end{proof}


\section{Short $A$-paths}
\label{sec:short}

In this additional section, we consider another variant of \textsc{$A$-Path Packing},
which we call \textsc{Short $A$-Path Packing}.
We call a nontrivial $A$-path of length at most $\ell$ an \emph{$A_{\le \ell}$-path}.
Now the problem considered here is defined as follows:
\begin{myproblem}
  \problemtitle{\textsc{Short $A$-Path Packing} \ (\sapp)}
  \probleminput{A tuple $(G,A,k,\ell)$, where $G = (V,E)$ is a graph, $A \subseteq V$, 
    and $k$ and $\ell$ are positive integers.}
  \problemquestiontitle{Question}
  \problemquestion{Does $G$ contain $k$ vertex-disjoint $A_{\le \ell}$-paths?}
\end{myproblem}

Similarly to {\falpp},
we define {\fsapp} as the restricted version of {\sapp} with $k = |A|/2$.

The positive results on {\alpp} presented so far
can be translated to the ones on {\sapp} by the following lemma.
\begin{lemma}
\label{lem:sapp-to-alpp}
Given an instance $(G,A,k,\ell)$ of {\sapp} where $G$ has $n$ vertices and $m$ edges,
one can compute an equivalent instance $(G',A,k,\ell)$ of {\alpp} in $O(m n^{2})$ time, 
where $G'$ has $O(m n^{2})$ vertices and edges,
and $\tw(G') \le \tw(G) + 1$.
\end{lemma}
\begin{proof}
We construct $G' = (V',E')$ from $G = (V,E)$ by replacing each edge $\{u,v\} \in E$ 
with $\ell$ new $u$--$v$ paths of lengths $1, 2, \dots, \ell$. (See \figref{fig:sapp}.)
We call these paths (including the one of length $1$) the \emph{detours} of $\{u,v\}$.
Clearly, $|V'|, |E'| \in O(m n^{2})$,
and $G'$ can be constructed in time linear in $|V'| + |E'|$.

\begin{figure}[bth]
  \centering
  \includegraphics[scale=1]{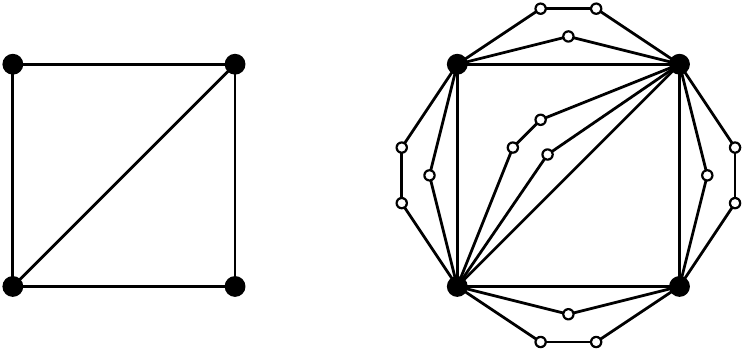}
  \caption{The construction of $G'$ (right) from $G$ (left) when $\ell = 3$.}
  \label{fig:sapp}
\end{figure}

To bound the treewidth of $G'$, 
observe that $G'$ can be seen as a graph obtained from $G$ by attaching triangles to edges,
and then by subdividing some edges.
Such operations preserve the treewidth unless $G$ is a forest.
(If $G$ is a forest, then the treewidth increases by $1$.)
To see this, let $\{u,v\}$ be the target edge of one of such operations.
Let $w$ be the new vertex introduced by the operation.
Every tree decomposition $\mathcal{T}$ of the original graph has a bag $B$ including both $u$ and $v$ 
since $\{u,v\}$ is an edge of the original graph.
We add to $\mathcal{T}$ a new bag $B' = \{u,v,w\}$ adjacent only to $B$.
Clearly, the obtained decomposition is a tree decomposition of the graph obtained by the operation,
and its width is the maximum of $|B'|-1$ and the width of $\mathcal{T}$.

\smallskip

Now assume that $(G,A,k,\ell)$ is a yes instance of {\sapp},
and let $P_{1}, \dots, P_{k}$ be vertex-disjoint $A_{\le \ell}$-paths in $G$.
For each $i \in [k]$, let $\ell_{i}$ be the length of $P_{i}$.
We construct an $(A,\ell)$-path $P_{i}'$ in $G$ from $P_{i}$
by replacing an arbitrary edge $\{u,v\}$ in $P_{i}$ with its detour of length $\ell - \ell_{i} + 1$.
Since $P_{1}, \dots, P_{k}$ are vertex-disjoint and the detours are internally vertex-disjoint,
the $(A,\ell)$-paths $P_{1}', \dots, P_{k}'$ are vertex-disjoint.

Conversely, assume that $(G',A,k,\ell)$ is a yes instance of {\alpp},
and let $P_{1}', \dots, P_{k}'$ be vertex-disjoint $(A,\ell)$-paths in $G'$.
Note that each $P_{i}'$ is a concatenation of some detours.
We obtain an $A_{\le \ell}$-path $P_{i}$ in $G$ by replacing all detours in $P_{i}'$ with the original edges in $G$.
Since $P_{1}', \dots, P_{k}'$ are vertex-disjoint and $V(P_{i}) \subseteq V(P_{i}')$ for each $i \in [k]$,
the $A_{\le \ell}$-paths $P_{1}, \dots, P_{k}$ are vertex-disjoint.
\qed
\end{proof}

By Lemma~\ref{lem:sapp-to-alpp},
Theorems~\ref{thm:ell<=3}, \ref{thm:k+ell}, \ref{thm:tw-XP}, \ref{thm:tw+ell} imply
the following positive results on {\sapp}.

\begin{corollary}
If $\ell \le 3$, then {\sapp} can be solved in polynomial time. 
\end{corollary}

\begin{corollary}
{\sapp} parameterized by $k + \ell$ is fixed-parameter tractable. 
\end{corollary}

\begin{corollary}
{\sapp} can be solved in time $n^{O(\tw)}$. 
\end{corollary}

\begin{corollary}
{\sapp} parameterized by $\tw + \ell$ is fixed-parameter tractable. 
\end{corollary}

On the other hand, the negative results on {\alpp} cannot be directly translated to the one on {\sapp}.
We here prove the hardness of the cases with constant $|A| \ge 4$ or with constant $\ell \ge 4$.
Note that {\sapp} with $k=1$ (or $|A| = 2$) is polynomial-time solvable because
it reduces to the all-pairs shortest path problem.
We leave the complexity of {\sapp} parameterized by $\tw$ unsettled.
\begin{theorem}
For every even constant $\alpha \ge 4$,
{\fsapp} with $|A| = \alpha$ is NP-complete.
\end{theorem}
\begin{proof}
Since the problem is clearly in NP, we present a reduction from the following NP-complete problem
2D1SP~\cite{Eilam-Tzoreff98}.
Given a graph $G = (V,E)$ and two terminal pairs $(s_{1},t_{1})$ and $(s_{2},t_{2})$,
the problem 2D1SP asks whether there exist two vertex-disjoint paths 
$P_{1}$ from $s_{1}$ to $t_{1}$ and $P_{2}$ from $s_{2}$ to $t_{2}$,
where $P_{1}$ is asked to be a shortest $s_{1}$--$t_{1}$ path in $G$.
We reduce 2D1SP to {\fsapp} with $|A| = 4$.
(We can extend this result to any even $\alpha = |A|$
by adding dummy components as we did in the proof of Observation~\ref{obs:fixed-|A|}.)

Let $n = |V|$ and $\ell = 5n-1$.
Let $\ell_{1}$ be the shortest path distance between $s_{1}$ and $t_{1}$ in $G$.
We add four vertices $s_{1}'$, $t_{1}'$, $s_{2}'$, and $t_{2}'$ to $G$.
We add 
a path of length $3n$ between $s_{1}$ and $s_{1}'$,
a path of length $2n-\ell_{1}-1$ between $t_{1}$ and $t_{1}'$,
a path of length $2n$ between $s_{2}$ and $s_{2}'$, and
a path of length $2n$ between $t_{2}$ and $t_{2}'$.
We call the obtained graph $G'$
and set $A = \{s_{1}', t_{1}', s_{2}', t_{2}'\}$.

Assume that $G$ has a shortest $s_{1}$--$t_{1}$ path $P_{1}$ 
and a (not necessarily shortest) $s_{2}$--$t_{2}$ path $P_{2}$ vertex-disjoint from $P_{1}$.
For each $i \in \{1,2\}$, we extend $P_{i}$ to a path $P_{i}'$ between $s_{i}'$ and $t_{i}'$
by adding the unique paths between $s_{i}'$ to $s_{i}$ an $t_{i}$ to $t_{i}'$.
The length of $P_{1}'$ is $3n + \ell_{1} + (2n-\ell_{1}-1) = \ell$ and 
the length of $P_{2}'$ is $2n + \|P_{2}\| + 2n \le \ell$,
where $\|P_{2}\|$ is the length of $P_{2}$.

Conversely, assume that $G'$ has two vertex-disjoint $A_{\le \ell}$-paths $P_{1}'$ and $P_{2}'$.
Without loss of generality, we can assume that one of the endpoints of $P_{1}'$ is $s_{1}'$.
Then we can see that the other endpoint of $P_{1}'$ is $t_{1}'$
since the distance between $s_{1}'$ and the vertices $s_{2}'$ and $t_{2}'$
is at least $3n + 2n > \ell$.
Let $P_{1}$ be the subpath of $P_{1}'$ that connects $s_{1}$ and $t_{1}$.
Now the length of the $A_{\le \ell}$-path $P_{1}'$ is $3n + \|P_{1}\| + (2n-\ell_{1}-1) \le \ell$,
and thus $\|P_{1}\| \le \ell_{1}$. Hence $P_{1}$ is a shortest path between $s_{1}$ and $t_{1}$ in $G$.
Let $P_{2}$ be the subpath of $P_{2}'$ that connects $s_{2}$ and $t_{2}$.
Since $P_{1}'$ and $P_{2}'$ are vertex-disjoint, so are $P_{1}$ and $P_{2}$.
\qed
\end{proof}

\begin{theorem}
For every constant $\ell \ge 4$, {\fsapp} is NP-complete. 
\end{theorem}
\begin{proof}
We show the NP-hardness of {\fsapp} with constant $\ell \ge 4$
by a reduction from a variant of 3-\textsc{Sat} with the following restrictions:
(1) each clause is a disjunction of two or three literals, and 
(2) each variable occurs exactly twice as a positive literal and exactly once as a negative literal.
We call this variant {\bsat}. It is known  that {\bsat} is NP-complete~\cite{FellowsKMP95}.

Let $(U, \mathcal{C})$ be an instance of {\bsat} with
the variables $U = \{u_{1}, \dots, u_{n}\}$ and the clauses $\mathcal{C} = \{C_{1}, \dots, C_{m}\}$.
If the positive literal of $u_{i}$ appears in $C_{p}$ and $C_{q}$ with $p < q$,
then we say that the first occurrence of $u_{i}$ is in $C_{p}$ and the second is in $C_{q}$.

For each $i \in [n]$, we construct the variable gadget for $u_{i}$ as follows (see \figref{fig:sapp4} (left)).
Take three paths of length $\ell$ from $s_{i}^{1}$ to $t_{i}^{1}$, from $s_{i}^{2}$ to $t_{i}^{2}$,
and from $\overline{s}_{i}$ to $\overline{t}_{i}$. We call these paths \emph{vertical}.
Add three paths of length $\ell-1$
from $s_{i}^{1}$ to the neighbor of $\overline{t}_{1}$,
from $s_{i}^{2}$ to the neighbor of $t_{i}^{1}$, and
from $\overline{s}_{i}$ to the neighbor of $t_{i}^{2}$.
Now these paths together with the edges incident to $t_{i}^{1}$, $t_{i}^{2}$, and $\overline{t}_{i}$
form three paths of length $\ell$ from $s_{i}^{1}$ to $\overline{t}_{1}$,
from $s_{i}^{2}$ to $t_{i}^{1}$, and from $\overline{s}_{i}$ to $t_{i}^{2}$.
We call these paths \emph{slanted}.
For $j \in \{1,2\}$, we call the vertex of distance $2$ from $t_{i}^{j}$ on the corresponding vertical path $x_{i}^{j}$.
We call the vertex of distance $2$ from $\overline{t}_{i}$ on the corresponding slanted path $\overline{x}_{i}$.

For each $C \in \mathcal{C}$, we construct the clause gadget for $C$ as follows (see \figref{fig:sapp4} (right)).
Assume that $C$ includes $c$ literals.
Take two vertices $s_{C}$ and $t_{C}$ and then
add $c$ internally disjoint paths of length $\ell$ between $s_{C}$ and $t_{C}$.
We bijectively map the neighbors of $t_{C}$ to the literals in $C$.
For each neighbor $v$ of $t_{C}$, 
if $v$ is mapped to the $j$th positive occurrence of $u_{i}$,
then we identify $v$ with $x_{i}^{j}$ in the variable gadget for $u_{i}$.
Similarly, if $v$ is mapped to the negative occurrence of $u_{i}$, then we identify $v$ with $\overline{x}_{i}$.

We call the constructed graph $G$ and
set $A = \{s_{i}^{1}, s_{i}^{2}, \overline{s}_{i}, t_{i}^{1}, t_{i}^{2}, \overline{t}_{i} \mid i \in [n]\}
\cup \{s_{C}, t_{C} \mid C \in \mathcal{C}\}$. This completes the construction.
We show that $(U, \mathcal{C})$ is a yes instance of {\bsat}
if and only if $(G,A,|A|/2, \ell)$ is a yes instance of {\fsapp}.

\begin{figure}[bth]
  \centering
  \includegraphics[scale=1]{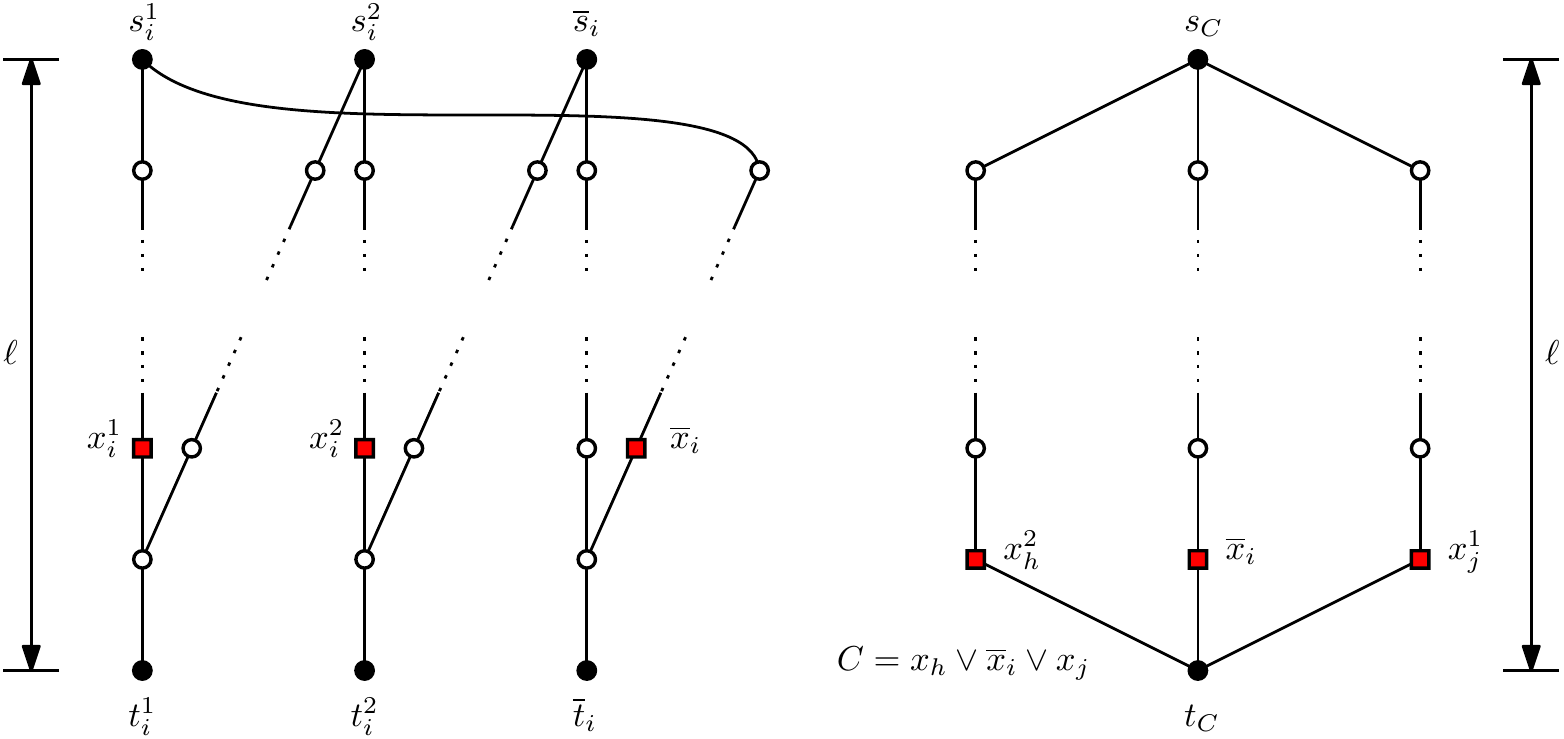}
  \caption{The variable gadget (left) and the clause gadget (right).}
  \label{fig:sapp4}
\end{figure}

To show the only-if direction,
assume that there is a truth assignment to the variables in $U$ that satisfies $\mathcal{C}$.
For each variable $u_{i} \in U$, if $u_{i}$ is set to be true,
then we take the slanted paths in the variable gadget for $u_{i}$;
otherwise, we take the vertical paths.
Since the variable gadgets are vertex-disjoint, the paths taken so far are vertex-disjoint.
Then for each clause $C \in \mathcal{C}$, let $l$ be a literal in $C$ that set to be true.
Observe that the neighbor, say $v_{l}$, of $t_{C}$ mapped to $l$ is not used 
in the paths selected in the variable gadgets.
Thus we can take the $s_{C}$--$t_{C}$ path passing through $v_{l}$.
Since all paths selected have length $\ell$
and all vertices in $A$ are used as endpoints of the selected path,
$(G,A,|A|/2, \ell)$ is a yes instance of {\fsapp}.

To prove the if direction,
assume that there is a set of $|A|/2$ vertex-disjoint $A_{\ell}$-paths $\mathcal{P}$ in $G$.
First observe that for each $C \in \mathcal{C}$, $t_{C}$ is the only vertex distance at most $\ell$ from $s_{C}$:
for each neighbor $v$ of $t_{C}$, the distance from $s_{C}$ to $v$ is $\ell - 1$,
and the distance from $v$ to any vertex in $A \setminus \{s_{C}, t_{C}\}$ is at least $\min\{2,\ell-2\} \ge 2$.
Thus, for each $C \in \mathcal{C}$, there is a path in $\mathcal{P}$ that has $s_{C}$ and $t_{C}$ as its endpoints.
Next we claim that for each variable $u_{i} \in U$,
either all vertical paths or all slanted paths in the variable gadget for $u_{i}$ are selected into $\mathcal{P}$.
To see this, observe that $\{\overline{t}_{i}, t_{i}^{1}\}$
(resp.~$\{t_{i}^{1}, t_{i}^{2}\}$, $\{t_{i}^{2}, \overline{t}_{i}\}$)
is the set of vertices in $A \setminus \{s_{C}, t_{C} \mid C \in \mathcal{C}\}$
that are distance at most $\ell$ from $s_{i}^{1}$ (resp.~$s_{i}^{2}$, $\overline{s}_{i}$).
Therefore, if we pick a vertical (resp.~slanted) path in a variable gadget,
then we have to take all vertical (resp.~slanted) paths in that variable gadget.
We now construct a truth assignment to $U$ by setting $u_{i}$ true if and only if
$\mathcal{P}$ includes the slanted paths in the variable gadget for $u_{i}$.
For $C \in \mathcal{C}$, let $P \in \mathcal{P}$ be the path connecting $s_{C}$ and $t_{C}$.
Let $l$ be the literal in $C$ corresponding to the neighbor of $t_{C}$ on $P$.
If $l$ is a positive literal of a variable $u_{i}$,
then $\mathcal{P}$ includes the slanted paths in the variable gadget for $u_{i}$,
and thus $u_{i}$ is set to be true.
If $l$ is a negative literal of $u_{i}$,
then $\mathcal{P}$ includes the vertical paths in the variable gadget for $u_{i}$,
and thus $u_{i}$ is set to be false.
In both cases, $l$ is true and $C$ is satisfied.
\qed
\end{proof}


\section{Concluding remarks}
\label{sec:conclusion}

In this paper, we have introduced a new problem \textsc{$(A,\ell)$-Path Packing} (\alpp)
and showed tight complexity results.
One possible future direction would be the parameterization by clique-width $\cw$,
a generalization of treewidth~(see \cite{HlinenyOSG08}).
In particular, we ask the following two questions.
\begin{itemize}
  \item Does {\alpp} admit an algorithm of running time $O(n^{\cw})$?
  \item Is {\alpp} fixed-parameter tractable parameterized by $\cw + \ell$?
\end{itemize}

We also considered a variant of the problem which we call \textsc{$A_{\le \ell}$-Path Packing} (\sapp).
We showed that results similar to the ones on {\alpp} hold also on {\sapp},
but we were not able to determine the complexity parameterized by treewidth. 
We left the following question on {\sapp}
\begin{itemize}
  \item Is {\sapp} W[1]-hard parameterized by $\tw$, $\tw + k$, or $\tw + |A|$?
\end{itemize}


\bibliographystyle{plainurl}
\bibliography{lap}

\end{document}